\newif\ifabstract

\abstractfalse 
\newif\iffull
\ifabstract \fullfalse \else \fulltrue \fi

\ifabstract
\documentclass{llncs}
\usepackage{llncsdoc}
\fi
\iffull
\documentclass[11pt]{article}
\fi
\usepackage[utf8]{inputenc}
\usepackage{tikz}
\usepackage{stmaryrd}
\usepackage{amsfonts}
\usepackage{amsmath}
\usepackage{amssymb}
\usepackage{amstext}
\usepackage{times}
\usepackage{xspace}
\usepackage{ifthen}
\usepackage{hyperref}

\usepackage[T1]{fontenc}
\usepackage{tikz}
\usetikzlibrary{shapes,shadows,arrows,patterns,trees,automata,positioning,snakes}
\usetikzlibrary{shapes.geometric}

\iffull
\advance \topmargin by -0.7in
\advance \oddsidemargin by -0.8in
\newcommand{\myparskip}{3pt}
\parskip \myparskip
\fi

\newcommand{\point}[1]{\node[circle,inner sep = 0pt,minimum size =0pt]}
\newcommand{\node}[1]{\node[circle,inner sep = 0pt,minimum size =3pt,fill = #1]}

\definecolor{dgreen}{rgb}{0,0.6,0}
\definecolor{dred}{rgb}{0.6,0,0}
\definecolor{dblue}{rgb}{0,0,0.6}
\definecolor{lgray}{rgb}{0.8,0.8,0.8}
\definecolor{lblue}{rgb}{0.6,0.6,1}
\definecolor{lgreen}{rgb}{0.6,1,0.6}
\definecolor{lred}{rgb}{1,0.6,0.6}

\newcommand{\vertex}[1]{\node[circle,inner sep = 0pt,minimum size =2pt,fill = #1]}
\newcommand{\Vertex}[1]{\node[circle,inner sep = 0pt,minimum size =3pt,fill = #1]}

\newcommand{\nellipse}[1]{\node[ellipse,minimum width=1cm,minimum height=0.55cm,draw=black,fill=#1]}

\newlength{\legendelength}
\newlength{\legendemaxlength}
\newcommand{\legende}[1]{
\settowidth{\legendelength}{\small\emph{#1}}
\ifthenelse{\lengthtest{\legendelength<\legendemaxlength}}
{\caption{\small \emph{#1}}}
{\caption{\small \protect\parbox[t]{0.8\textwidth}{\emph{#1}}}}
}

\iffull
\usepackage{amsthm}
\theoremstyle{plain}
\newtheorem{theorem}{Theorem}[section]
\newtheorem{lemma}[theorem]{Lemma}
\theoremstyle{remark}

\newtheorem{remark}{Remark}
\fi

\newtheorem{prop}[theorem]{Proposition}
\newtheorem{conj}{Conjecture}

\newcommand{\MC}{\gamma}
\newcommand{\C}{M}

\newcommand{\eqdef}{:=}

\newcommand{\tree}{\mathcal{T}}
\newcommand{\opt}{\textrm{OPT}}

\newcommand{\edp}{{\sc EDP}\xspace}
\newcommand{\ndp}{{\sc NDP}\xspace}
\newcommand{\medp}{{\sc Max EDP}\xspace}

\newcommand{\val}[1]{\left|#1\right|}

\newcommand{\polylog}{\text{polylog}}
\newcommand{\etal}{et al.\ }
\newcommand{\cA}{\mathcal{A}}
\newcommand{\cG}{\mathcal{G}}
\newcommand{\cT}{\mathcal{T}}
\newcommand{\bU}{{\bar{U}}}
\newcommand{\bx}{\bar{x}}

\begin{document}
\title{Maximum Edge-Disjoint Paths in $k$-Sums of Graphs}

\iffull
\author{
Chandra Chekuri\thanks{Dept.\ of Computer Science, University
of Illinois, Urbana, IL 61801. {\tt chekuri@illinois.edu}.
Work on this paper is partly supported
by NSF grant CCF-1016684.}
\and
Guyslain Naves\thanks{Laboratoire d'Informatique Fondamentale,
Faculté des sciences de Luminy, Marseille, France.
{\tt guyslain.naves@lif.univ-mrs.fr}.
}
\and
F. Bruce Shepherd\thanks{Dept. Mathematics and Statistics, McGill University, Montreal. {\tt bruce.shepherd@mcgill.ca}. Work is partly supported by NSERC Discovery/Accelerator Grants.}
}
\date{\today}
\fi

\ifabstract
\author{Chandra Chekuri \and Guyslain Naves \and Bruce Shepherd}

\institute{
Dept.\ of Computer Science, University of Illinois, Urbana, IL 61801,
USA. \\ \email{chekuri@illinois.edu}. 
\and
Laboratoire d'informatique fondamentale, Faculté des sciences de Luminy,
Marseille, France\\
\email{guyslain.naves@lif.univ-mrs.fr}
\and
Dept. Mathematics and Statistics, McGill University, Montreal, Canada \\
\email{bshepherd@math.mcgill.ca}
}
\fi

\maketitle

\begin{abstract}
  We consider the approximability of the maximum edge-disjoint paths
  problem (MEDP) in undirected graphs, and in particular, the
  integrality gap of the natural multicommodity flow based relaxation
  for it. The integrality gap is known to be $\Omega(\sqrt{n})$ even
  for planar graphs \cite{GargVY97} due to a simple topological
  obstruction and a major focus, following earlier work
  \cite{KleinbergT98}, has been understanding the gap if some constant
  congestion is allowed. In planar graphs the integrality gap is
  $O(1)$ with congestion $2$ \cite{SeguinS11,CKS-planar-constant}.  In
  general graphs, recent work has shown the gap to be
  $\polylog(n)$ \cite{Chuzhoy12,ChuzhoyL12} with congestion $2$.
  Moreover, the gap is $\log^{\Omega(c)} n$ in general graphs with
  congestion $c$ for any constant $c \ge 1$ \cite{andrews2010inapproximability}.

  It is natural to ask for which classes of graphs does a
  constant-factor constant-congestion property hold. It is easy to
  deduce that for given constant bounds on the approximation and
  congestion, the class of ``nice'' graphs is minor-closed. Is the
  converse true? Does every proper minor-closed family of graphs
  exhibit a constant factor, constant congestion bound relative to the
  LP relaxation? We conjecture that the answer is yes.
  One stumbling block has been  that such
  bounds were not known for bounded treewidth graphs (or even treewidth $3$). In this paper we
  give a polytime algorithm which takes a fractional routing solution
  in a graph of bounded treewidth and is able to integrally route a
  constant fraction of the LP solution's value. Note that we do not incur any edge congestion. Previously this was
  not known even for series parallel graphs which have treewidth $2$.
  The algorithm is based on a more general argument that applies to
  $k$-sums of graphs in some graph family, as long as the graph family
  has a constant factor, constant congestion bound. We then use this to show
  that such bounds hold for the class of $k$-sums of bounded genus
  graphs.
\end{abstract}

\section{Introduction}
The disjoint paths problem is the following: given an undirected graph
$G=(V,E)$ and node pairs $H = \{s_1t_1,\ldots,s_pt_p\}$, are there
disjoint paths connecting the given pairs? We use \ndp and \edp to
refer to the version in which the paths are required to be node-disjoint or edge-disjoint. Disjoint path problems are cornerstone
problems in combinatorial optimization. The seminal work on graph minors
of Robertson and Seymour \cite{robertson1985graph} gives a polynomial
time algorithm for \ndp (and hence also for \edp) when $p$ is fixed;
the algorithmic and structural tools developed for this have led to
many other fundamental results.  In contrast to the undirected case,
the problem in directed graphs is NP-Complete for $p=2$
\cite{FortuneHW1980}.  Further, \ndp and \edp are NP-Complete
in undirected graphs when $p$ is part of the input.
The maximization versions of \edp and
\ndp have also attracted intense interest, especially in connection to its
approximability.  In the {\em maximum edge-disjoint path} (\medp)
problem we are given an undirected (in this paper) graph $G=(V,E)$,
and node pairs $H=\{s_1t_1,\ldots ,s_pt_p\}$ , called
\emph{commodities} or {\em demands}. \medp \ asks for a maximum size
subset $I \subseteq \{1,2, \ldots ,p\}$ of commodities which is {\em routable}. A
set $I$ is routable if there is a family of edge-disjoint paths
$(P_i)_{i \in I}$ where $P_i$ has extremities $s_i$ and $t_i$ for each
$i \in I$. In a more general setting, the edges have integer capacities
$c: E(G) \to \mathbb{N}$, and instead of edge-disjoint paths, we ask
that for each edge $e \in E(G)$, at most $c(e)$ paths of $(P_i)_{i \in
  I}$ contain $e$.
For any demand $h = st
\in H$, denote by $\mathcal{P}_h$ the set of $st$-paths in $G$, and
$\mathcal{P} = \bigcup_{h \in H} \mathcal{P}_h$.
A natural linear
programming relaxation of \medp{} is then:

\begin{equation}
\begin{array}{lll}
\max & \sum_{h \in H} z_h  &\;\textrm{subject to}         \\
     & \sum_{P \in \mathcal{P}_h} x_P = z_h \leq 1        &\qquad\textrm{(for all $h \in H$)} \\
     & \sum_{P \in \mathcal{P}, e \in P} x_P \leq c_e &\qquad\textrm{(for all $e \in E$)} \\
     & x \geq 0
\end{array}
\label{eqn:LP}
\end{equation}

NP-Completeness of \edp implies that \medp is NP-Hard. In fact, \medp is
NP-Hard in capacitated trees for which \edp is trivially solvable. This indicates that \medp inherits
hardness also from the selection
of the subset of demands to route.  As pointed out in \cite{GargVY97}, a grid
example shows that the integrality gap of the multicommodity flow
relaxation may be as large as $\Omega(\sqrt{n})$ even in planar
graphs.  However, the grid example is not robust in the sense that if
we allow edge-congestion $2$ (or equivalently, if we assume all
capacities are initially at least $2$), then the example only has a
constant factor gap. This observation led Kleinberg-Tardos
\cite{KleinbergT98} to seek better approximations
(polylog or constant factor) for planar graphs in the regime where
some low congestion is allowed. With some work, this agenda proved
fruitful: a constant-approximation with edge congestion $4$ was proved
possible in planar graphs \cite{CKS-planar-constant}; this was improved to (an
optimal) edge congestion $2$ in \cite{SeguinS11}.

In general graphs,   Chuzhoy \cite{Chuzhoy12} recently obtained the first
poly-logarithmic approximation with constant congestion ($14$).
This was subsequently improved to the optimal congestion of $2$ by
Chuzhoy and Li \cite{ChuzhoyL12}.  It is also known that, in general
graphs, the integrality gap of the flow LP is
$\Omega(\log^{\Omega(1/c)} n)$ even if congestion $c$ is allowed; the
known hardness of approximation results for \medp with congestion have
similar bounds as the integrality gap bounds, see
\cite{andrews2010inapproximability}.

For any constants $\alpha,\beta \geq 1$, one may ask for which graphs does
the LP for \medp admit an integrality gap of $\alpha$ if edge
congestion $\beta$ is allowed.  It is natural to require this for any
possible collection of demands and any possible assignment of edge
capacities. For fixed constants, it is easy to see that the class of
such graphs is closed under minors.
Is the converse true? That is, do all minor-closed graphs exhibit a {\em
  constant factor constant-congestion} (CFCC) integrality gap for
\medp? In fact we consider the following stronger conjecture with congestion
$2$.

\begin{conj}
  \label{conj:minor-free}
  Let ${\cal G}$ be any proper minor-closed family of graphs. Then the
  integrality gap of the flow LP for \medp is at most a constant $c_{\cal G}$
  when congestion $2$ is allowed.
\end{conj}

The preceding conjecture is inherently a geometric question, but one
would also anticipate a polytime algorithm for producing the routable
sets which establish the gap. In attempting to prove
Conjecture~\ref{conj:minor-free}, one must delve into the
structure of minor-closed families of graphs, and in particular the
characterization given by Robertson and Seymour \cite{robertson1985graph}. Two
minor-closed families that form the building blocks for this
characterization are (i) graphs embedded on surfaces of bounded genus (in
particular planar graphs), and (ii) graphs with bounded treewidth.
For \medp, we have a constant factor integrality
gap with congestion $2$ for planar graphs. In \cite{CKS-treewidth} it
is shown that the integrality gap of the LP for \medp in graphs of
treewidth at most $k$ is $O(k \log k \log n)$; note that this is with
congestion $1$. Existing integrality gap results, when interpreted in
terms of treewidth $k$, show that the integrality gap is $\Omega(k)$ for
congestion $1$ and $\Omega(\log^{O(1/c)} k)$ for congestion $c >
1$. It was asked in \cite{CKS-treewidth} whether the gap is
$O(k)$ with congestion $1$. In particular, the question of
whether the gap is $O(1)$ for $k=2$ (this is precisely the
class of series parallel graphs) was open.
In this paper we show the following result.

\begin{theorem}
  \label{thm:tw}
  The integrality gap of the flow LP for \medp is $2^{O(k)}$ in graphs
  of treewidth at most $k$. Moreover, there is a polynomial-time
  algorithm that given a graph $G$, a tree decomposition
  for $G$ of width $k$, and fractional solution to the LP of value
  ${\sf OPT}$, outputs an integral solution of value $\Omega({\sf OPT}/2^{O(k)})$.
\end{theorem}

The preceding theorem is a special case of a more general theorem that
we prove below. Let $\cG$ be a family of graphs. For any integer $k \ge 1$,
let $\cG_k$ denote the class of graphs obtained from $\cG$ by the
$k$-sum operation. The $k$-sum operation is formally defined in
Section~\ref{sec:ksums-defn}; the structure theorem of Robertson and Seymour
is based on the $k$-sum operation over certain classes of graphs.

\begin{theorem}
\label{thm:ksums}
  Let $\cG$ be a minor-closed class of graphs such that
  the integrality gap of the flow LP is $\alpha$ with congestion $\beta$.
  Then the integrality gap of the flow LP for the class $\cG_k$ is
  $2^{O(k)} \alpha$ with congestion $\beta+3$.
\end{theorem}

The preceding theorem is effective in the following sense: there is a
polynomial-time algorithm that gives a constant factor, constant
congestion result for $\cG_k$ assuming that (i) such an algorithm
exists for $\cG$ and (ii) there is a polynomial-time algorithm to find
a tree decomposition over $\cG$ for a given graph $G \in \cG_k$.

We give the following as a second piece of evidence towards
Conjecture~\ref{conj:minor-free}.

\begin{theorem}
  \label{thm:genus}
  The integrality gap of the flow LP on graphs of genus $g > 0$
  is $O(g \log^2 (g+1))$ with congestion $3$.
\footnote{
  We believe that the congestion bound in the preceding theorem
  can be improved to $2$ with some additional technical work.
  We do not give a polynomial-time algorithm although we believe
  that it too is achievable with some (potentially messy) technical work.}
  \end{theorem}

Theorems~\ref{thm:ksums} and \ref{thm:genus} imply that the class
of graphs obtained as $k$-sums of graphs with genus $g$
is CFCC when $k$ and $g$ are fixed constants. The bottleneck
in extending our results to prove Conjecture~\ref{conj:minor-free}
are planar graphs (or more generally bounded genus graphs)
that have ``vortices'' which play a non-trivial role in the
Robertson-Seymour structure theorem.




\paragraph{A brief discussion of technical ideas and related work:}
The approximability of \medp in undirected and directed graphs has
received much attention in the recent years. We refer the reader to
some recent papers
\cite{ChuzhoyL12,Chuzhoy12,SeguinS11,andrews2010inapproximability}. A
framework based on well-linked decompositions \cite{CKS-well-linked}
has played an important role in understanding the integrality gap of
the flow relaxation in undirected graphs.  It is based on recursively
cutting the input graph along sparse cuts until the given instance is
well-linked. However, this framework loses at least a logarithmic
factor in the approximation. The work in \cite{CKS-planar-constant}
obtained a constant factor approximation for planar graphs by using a
more refined decomposition that took advantage of the structure of
planar graphs. For graphs of treewidth $k$, \cite{CKS-treewidth} used the
well-linked decomposition framework to obtain an $O(k \log k \log
n)$-approximation and integrality gap. Our work here shows that one
can bypass the well-linked decomposition framework for bounded
treewidth graphs, and more generally for $k$-sums over families of
graphs.  The key high-level idea is to effectively reduce the
(tree)width of one side of a sparse cut if the terminals cannot route
to a small set of nodes. Making this work requires a somewhat nuanced induction
hypothesis.  For bounded-genus graphs, we adapt the well-linked
decomposition to effectively reduce the problem to the planar graph
case.

There are two streams of questions comparing minimum cuts to maximum
flows in graphs.  First, the {\em flow-cut gap} measures the gap
between a sparsest cut and a maximum concurrent flow of an instance.
The second measures the {\em throughput-gap} by comparing the maximum
throughput flow and the minimum multicut.  These gap results have been
of fundamental importance in algorithms starting with the seminal work
of Leighton and Rao \cite{LeightonR99}. It is known that the gaps in
general undirected graphs are $\Theta(\log n)$; see
\cite{Shmoys-survey} for a survey.  It is also conjectured
\cite{gupta2004cuts} (the GNRS Conjecture) that the flow-cut gap is
$O(1)$ for minor-closed families.  This conjecture is very much open
and is not known even for planar graphs or treewidth $3$ graphs; see
\cite{lee2010genus} for relevant discussion and known results.  In
contrast, the work of Klein, Plotkin and Rao \cite{KPR} showed that
the throughput-gap is $O(1)$ in any proper minor-closed family of
graphs (formally shown in \cite{TardosV93}).  The focus of these works
is on fractional flows, in contrast to our focus on integral routings.
Conjecture~\ref{conj:minor-free} is essentially asking about the
integrality gap of throughput flows.  Given the $O(1)$ throughput-gap
\cite{KPR}, it can also be viewed as asking whether the gap between
the maximum {\em integer} throughput flow with congestion $2$ is
within an $O(1)$ factor of the minimum multicut. Analogously for
flow-cut gaps, \cite{ChekuriSW} conjectured that the gap between the
maximum integer concurrent flow and the sparsest cut is $O(1)$ in
minor-free graphs.

\ifabstract
\paragraph{Organization:} 
Due to space constraints several technical ingredients needed 
to prove Theorem~\ref{thm:tw}
and Theorem~\ref{thm:ksums} are in the
appendix; a summary of these ingredients at a high-level is provided
in Section~\ref{sec:tools-summary} after some preliminaries.
The proof of Theorem~\ref{thm:genus} and discussion of open problems
are also moved to the appendix. 

\fi
\section{Preliminaries}

Recall that an instance of \medp consists of a graph $G$ and demand
pairs $H$. In general $H$ can be a multiset, however it is convenient
to assume that $H$ is a matching on the nodes of $G$. Indeed we just
have to attach the terminals to leaves created from new nodes.  With
this assumption we use $X$ to denote the set of terminals (the endpoints
of the demand pairs) and $M$ the matching on $X$ that
corresponds to the demands. We call the triple $(G,X,M)$ a matching
instance of \medp.  Let $\bx$ be a feasible solution to the LP
relaxation (\ref{eqn:LP}).  For each node $v \in X$, we also use
$x(v)$ to denote the value $\sum_{P \in \mathcal{P}_h} x_P$ where $v$
is an endpoint of the demand $h$; this is called the {\em
  marginal value} of $v$.  We  assume that all capacities $c_e$
are $1$; this does not affect the integrality gap analysis. Moreover,
as argued previously (cf. \cite{Chekuri04b}), at a loss of a factor of
$2$ in the approximation ratio, the assumption can be made for
polynomial-time algorithms that are based on rounding a solution to
the flow relaxation.

\subsection{$k$-Sums and the structure theorem of Robertson and Seymour}\label{sec:ksums-defn}

Let $G_1$ and $G_2$ be two graphs, and $C_i$ a clique of size $k$ in
$G_i$. The graph $G$ obtained by identifying the nodes of $C_1$
one-to-one with those of $C_2$, and then removing some of the edges
between  nodes of $C_1 = C_2$, is called a \emph{$k$-sum} of $G_1$
and $G_2$.  For a class of graphs $\cG$, we define the class $\cG_k$ of
the graphs {\em obtained from $\cG$ by $k$-sums}, to be the smallest class
of graphs such that: (i) $\cG$ is included in $\cG_k$, and (ii) if $G$
is a $k$-sum of $G_1 \in \cG$ and $G_2 \in \cG_k$, then $G \in \cG_k$.

Fix a class of graphs $\cG$. A tree $\cT$ is a {\em tree decomposition} over
$\cG$ for a graph $G = (V,E)$, if each node $A$ in $\cT$ is associated
to a subset of nodes $X_A \subseteq V$, called a \emph{bag}, and
the following properties hold:
\begin{itemize}\setlength{\itemsep}{0pt}
\item[$(i)$] for each $v \in V(G)$, the set of nodes of $\cT$ whose
  bags contain $v$, form a non-empty sub-tree of $\cT$,
\item[$(ii)$] for each edge $uv \in E(G)$, there is a bag with both $u$ and $v$
  in it,
\item[$(iii)$] for any bag $X$, the graph obtained from $G[X]$ by adding cliques over $X \cap Y$, for every adjacent bag $Y$, is in $\cG$. We  denote this graph by $G[[X]]$.
\end{itemize}
When $\cG$ is closed under taking minors, condition $(iii)$ implies
that $G[X]$ itself is in $\cG$, as well as any graph obtained from
$G[X]$ by adding edges in $X \cap Y$, for any adjacent bag $Y$.
Throughout we assume that $\cG$ is minor-closed. We
sometimes identify the nodes of $\cT$ with their respective bags.  We
also denote by $V(\cT)$, the union of all bags, and so $V(\cT)
\subseteq V(G)$.

A set of nodes $X \cap Y$, for $X$ and $Y$ adjacent bags, is called a
\emph{separator}. When the tree decomposition $\cT$ is minimal (with
respect to the number of bags), the separators are disconnecting node
sets of $G$. Thus each edge $e$ of $\cT$ identifies a separator,
denoted by $V_e$.
For convenience, we usually work with rooted tree decompositions,
where an arbitrary node is chosen to be the root. Then, for bag $X$
and its parent $Y$, we denote by $S_X$ the separator $X \cap Y$.


The \emph{width} of a tree decomposition $\cT$ is the maximum
cardinality of a separator of $\cT$. The \emph{width} of a graph
(relative to a graph class $\cG$) is the smallest width of a tree
decomposition for that graph. A graph of width $k$ can thus be
obtained by $k$-sums of graphs from $\cG$.  As a special case, the
\emph{treewidth} of a graph $G$ is the smallest $k$ such that $G$
admits a decomposition of width $k$ relative to the class of all
graphs with at most $k+1$ nodes.


Let $\tree$ be a tree decomposition
of a graph $G$, rooted at a node $R$. For any edge $e$ of the tree decomposition,
 let $\overline{\tree_e}$ and $\tree_e$ be the
subtrees obtained from $\tree$ by removing $e$, with $R \in
\overline{\tree_e}$. We denote by $G_e$ the graph obtained from the
induced subgraph of $G$ on node set $V(\tree_e)$, and then removing all edges in
$V(\overline{\tree_e}) \times V(\overline{\tree_e})$. Note that
$\tree_e$ is a tree decomposition of $G_e$.

We recall informally the  graph structure theorem proved by Robertson and
Seymour. For $k \in \mathbb{N}$, let $\mathcal{L}^k$ be the graphs
obtained in the following way.
\begin{itemize}\setlength{\itemsep}{0pt}
\item we start from a graph $G$ embeddable on a surface of genus $k$,
\item then we add \emph{vortices} of width $k$ to at most $k$ faces of $G$,
\item then we add at most $k$ \emph{apex} nodes. That is, each of
  these nodes can be adjacent to an arbitrary subset of nodes.
\end{itemize}
Then, we denote $\mathcal{L}^k = \mathcal{L}^k_k$. For a graph $H$, we
denote by $\mathcal{K}_H$ the graphs that do not contain an $H$-minor.

\begin{theorem}[Robertson and Seymour~\cite{RobertsonS2003}]
  For any graph $H$, there is an integer $k > 0$ such that
  $\mathcal{K}_H \subseteq \mathcal{L}^k$.
\end{theorem}

In order to prove Conjecture~\ref{conj:minor-free},
one should be able to use the preceding decomposition theorem,
proving that the CFCC property holds for bounded genus graph and is
preserved by adding a constant number of vortices and apex nodes, and
by taking $k$-sums. Apex nodes are easy to deal with. This paper
provides a proof for bounded genus graphs and for $k$-sums. This
leaves only the cases of vortices as the bottleneck
in proving the conjecture.

\ifabstract
\section{Technical Ingredients}
\label{sec:tools-summary}
We rely on several technical tools and ingredients that are either explicitly
or implicitly used in recent work on \medp. Due to space constraints
they are in Section~\ref{sec:tools} of the appendix.
Here we give a high-level description.

\paragraph{Moving Terminals:} The idea here (also leveraged in
previous work, cf. \cite{CKS-planar-constant,CKS-treewidth}) is to
reduce a \medp instance to a simpler/easier instance by moving the terminals
to a specific set of new locations (nodes). The two instances are
equivalent for \medp, up to an additional constant congestion and
constant factor approximation. 

\paragraph{Sparsifiers:} Given a graph $G=(V,E)$ and $S \subset V$ a
sparsifier is a graph $H$ only on the node set $S$ that acts as a
proxy for routing between nodes in $S$ in the original graph $G$.
We are interested in integer sparsifiers of certain type when $|S|$
is small, a constant in our setting.

\paragraph{Routing through a small set of nodes:} The idea here is that
if all the flow for a given instance intersects a small set of nodes,
say $p$, then one can in fact obtain an $\Omega(1/p)$-approximation
for \medp. This follows the ideas from \cite{CKS-sqrtn} where the special
case of $p=1$ was exploited.

\fi

\iffull
\section{Technical Ingredients}
\label{sec:tools}
We rely on several technical tools and ingredients that are either explicitly
or implicitly used in recent work on \medp.

\subsection{Moving Terminals}
We first describe a general tool (ideas of which are leveraged also in
previous work, cf. \cite{CKS-planar-constant,CKS-treewidth}) that allows
us to reduce a \medp instance to a simpler one by
moving the terminals to a specific set of new locations (nodes). The
two instances are equivalent for \medp, up to an additional constant
congestion and constant factor approximation.


\begin{lemma}
\label{lem:rr}
Suppose we have a (matching) instance $(G,H)$ of \medp{} with some
solution $\bx$ to its LP relaxation. Let $\val{\bx} = \sum_P x_P$. Suppose
that for some $S,R \subseteq V(G)$, there is a flow which routes
$x(s)$ from each $s \in S$, and all flow terminates in $R$.  Then
there is another (matching) instance $(G',H')$ with the following
properties.
\begin{enumerate}\setlength{\itemsep}{0pt}
\item The new instance has a (fractional) solution of value at least
  $\val{\bx}/5$,
\item If there is an integral solution for $(G',H')$ of congestion $c$, 
  then there is an integral solution for $(G,H)$
  of the same value and congestion $c+2$,
\item $G',H'$ is obtained from $G,H$ by hanging off pendant stars from
  some of the nodes.
\end{enumerate}
\end{lemma}

\begin{proof}
Let $T$ be a forest of $G$ spanning all the nodes of $S$ such that
each component of $T$ contains at least one node of $R$. We consider
each component of $T$ separately, so we assume here that $T$ is a
tree. Let $r \in R \cap V(T)$ and take this as a root of $T$.

We partition $S$ into subsets $S_1,\ldots,S_\ell$ with the following properties.
\begin{itemize}\setlength{\itemsep}{0pt}
\item[$(i)$] for each $i \in [1,\ell]$, $1 \leq x(S_i) \leq 2$ (except
  possibly $S_1$ may have $x(S_1) < 1$),
\item[$(ii)$] to each $S_i$ is associated a subtree $T_i$ of $T$ spanning $S_i$,
\item[$(iii)$] $T_1,\ldots, T_l$ are edge-disjoint.
\end{itemize}
We achieve this by  the following iterative scheme.
\begin{itemize}\setlength{\itemsep}{0pt}
\item If $x(T) \leq 2$, choose $S_1 = S$, $T_1 = T$, $\ell = 1$. Else:
\item Find a deepest node $v$ in $T$, such that the subtree $T'$
  rooted at $v$ has $x(T') \geq 1$.
\item let $v_1,\ldots,v_i$ be the children of $v$, and $T'_1,\ldots,T'_i$
  be the subtrees rooted at $v_1,\ldots,v_i$ respectively. If
  $\sum_{j=1}^i x(T'_j) < 1$, then define $A \eqdef E(T')$ and $U
  \eqdef V(T') \cap S$.  Otherwise, find the smallest $i'$ such that
  $x(B) \geq 1$, where $B \eqdef \bigcup_{j=1}^{i'} V(T'_j) \cap S$.
  Then set $A \eqdef \bigcup_{j=1}^{i'} (E(T'_j) \cup \{vv_j\})$. In
  both cases, we get $1 \leq x(A) \leq 2$, and $A$ induces a tree.
  (Note that in the latter case $v$ was not placed in $B$ but does lie
  in $V(A)$.)
\item proceed inductively on $T - A$ and $S - B$, to find $S_1,\ldots
  S_{\ell'}$ and $T_1,\ldots T_{\ell'}$. Then $\ell \eqdef \ell' + 1$, $S_\ell \eqdef B$
  and $T_\ell \eqdef A$.
\end{itemize}
Note that with this scheme, $x(S_1)$ might be less than $1$, but in
that case $T_1$ contains the root node from $R$. As the hypotheses
$x(S_i) \geq 1$ is only used to prove the existence of edge-disjoint
paths from the $S_i$'s to $R$, this does not pose a problem (we can
simply take the trivial path at $r$ for $S_1$).

Each $S_i$ is called a \emph{cluster}. As each cluster sends at least
one unit of flow to $R$ (with the exception of $S_1$ already
mentioned), there is a family of edge-disjoint paths $P_1,\ldots,P_l$,
where each $P_i$ goes from a node $s_i \in S_i$ to some node $r_i \in
R$.  This can be seen by adding dummy source nodes (one for each
$S_i$) adjacent to nodes in each $S_i$, and a single dummy sink node
adjacent from each $r \in R$ (a detailed proof is found in \cite{CKS-planar-constant}).

 We now define a new instance of \medp{} $G',H'$. $G'$ is obtained
 from $G$ by adding $\ell$ new nodes $u_1,\ldots,u_\ell$ with degree one,
 where $u_i$ is adjacent to $r_i$. The capacity of a new edge $r_iu_i$
 is $1$, and we re-define $P_i$ as extending to $u_i$.  We identify
 each terminal in $S$ with the $u_i$ associated with its cluster as
 follows. Let $\phi(s) \eqdef s$ if $s \notin S$ and $\phi(s) = u_i$
 if $s \in S_i$. Then let $\phi(H) \eqdef\{\phi(s)\phi(t)~:~st \in
 H\}$. These demands do not yet form a matching, so $H'$ is obtained
 from $\phi(H)$ by simply deporting each of the terminals in $S$ to
 new nodes forming leaves.

  We show how to transform $\bx$ into a fractional flow $\bx'$ in $G',H'$
  with congestion $5$, such that $\bx'$ has the same value as $\bx$. For
  that, we only extend the flow paths for the demands in $H' \setminus
  H$. Let $st \in H$ be such a demand and $s't' = \phi(s)\phi(t)$ its
  image. For any $st$-path $P$ with value $x_P$, let $P'$ be the path
  obtained from $P$ by:
\begin{itemize}\setlength{\itemsep}{0pt}
\item if $s \in S_i$ for some $i$, concatenate $P_i$ and the unique $ss_i$-path of $T_i$,
\item similarly if $t \in S_j$ for some $j$, concatenate $P_j$ and the unique $ss_j$-path of $T_j$.
\end{itemize}

Then, set $x'_{P'} \eqdef x_P$. Then $\bx'$ has the same value as $\bx$ by
construction, but has higher congestion.  Any of the edges $r_iu_i$
(or additional leaves from an $x_i$) have congestion at most $2$ by
construction; thus it is enough to focus on edges within $G$. The
original flow paths incur congestion of at most 1 on any edge, so we
address the added congestion from extending the flow paths. The edges
of any $P_i$ are charged by at most $2$ units (by terminals within
$S_i$) and each $T_i$ is also charge by at most $2$ units. As an edge
may be contained in at most one $P_i$ and at most one $T_i$ the extra
congestion is bounded by $4$. Hence the total congestion of $x'$ is at
most $5$, and in particular, this implies that the fractional optimal
solution in $G'$ is at least $\frac{1}{5} \opt$.

 Suppose now that we have an integral solution to \medp{} for $G',H'$
 with congestion $c$. We show how to transform it into a solution for
 $G,H$ of the same value with congestion $c+2$. Since we can assume
 the flow paths used are simple, we only need to address the flow
 paths for demands in $H' \setminus H$. Let $P$ be any path in the
 solution satisfying a commodity associated with a node $s' =
 \phi(s)$, where $s$ is in $S_i$. Then we extend $P$ by concatenating
 $P_i$ and the unique $s_is$-path of $T_i$ to it.  We may also
 shortcut this to obtain a simple path. Again, this clearly defines a
 solution of same value to the original problem. Since the capacity of
 $r_iu_i$ is one, we use each $P_i$ at most once, and a path in $T_i$
 is used for only one such $s \in S_i$.  As the paths $P_i$ are
 disjoint, and the subtrees are disjoint, each edge is used at
 most $c+2$ times: $c$ for the original routing, $1$ for the $P_i$
 paths, and $1$ for the paths inside the $T_i$'s.
\end{proof}

\subsection{Sparsifiers}

Let $G=(V,E)$ be a (multi) graph and let $S \subset V$. We are
interested in creating a graph $H$ only on the node set $S$ that acts
as a proxy for routing between nodes in $S$ in the original graph $G$.
The notion of sparsifiers, introduced in
\cite{moitra2009approximation}, has many possible formulations depending on the  various applications.
For instance, a Gomory-Hu Tree can be viewed as a sparsifier which encodes pairwise maximum flows in a graph.
We are interested in the following model.
We
say that $H=(S,E_H)$ is a {\em $(\sigma,\rho)$-sparsifier} for $S$ in $G$ if
the following properties are true:
\begin{itemize}
\item any feasible (fractional) multicommodity flow in $G$ with the
  endpoints in $S$ is (fractionally) routable in $H$ with congestion
  at most $\sigma$
\item any integer multicommodity flow  in $H$ is integrally routable
  in $G$ with congestion $\rho$.
\end{itemize}

Existing sparsifier results mostly focus on fractional routing (or cut
preservation) while we need {\em integer} sparsifiers in the sense of
the second point above. Chuzhoy~\cite{Chuzhoy12} developed an integer
sparsifier result but it uses Steiner nodes and has limitations that
preclude its direct use in our setting. Instead, a simple
argument based on splitting-off 
gives the following weak
sparsifier result that suffices for our purposes.

\begin{theorem}
  \label{thm:sparsifier}
  Let $G=(V,E)$ be a graph and $S \subset V$. There
  is a $(|S|^2,2)$-sparsifier for $S$ in $G$.
\end{theorem}

\begin{proof}
  First, by standard $T$-join theory, $G$ contains a subset $E'$ of
  edges, such that if we add an extra copy of each such edge, we
  obtain an Eulerian graph $G'$.  We may now apply splitting off
  repeatedly at the (even degree) nodes of $V-S$.  Each operation
  preserves the minimum cut between any pair of nodes $u,v \in S$.
  This ultimately results in a (multi) graph $H=(S,F)$ on $S$.
  We claim that $H$ is the desired sparsifier.

  Note that any integral routing in $H$ can easily mapped to an
  integral routing in $G'$ since the edges in $F$ map to
  edge-disjoint paths in $G'$.  Since $G'$ had potentially an extra
  copy of an edge from $G$ we see that $\rho = 2$.

  Now consider any fractional multicommodity flow in $G$ between nodes
  in $S$. Say $d(uv)$ flow is routed between $u,v \in S$.  Then
  $d(u,v) \le \lambda_G(u,v)$ where $\lambda_G(u,v)$ is the capacity
  of a min $u$-$v$ cut in $G$. Since the splitting-off operation
  preserved connectivity, $\lambda_H(u,v) \geq \lambda_G(u,v)$, hence
  we can route $d(u,v)$ flow between $u$ and $v$ in $H$.
  However, we have $|S|(|S|-1)/2$ distinct pairs of nodes in $S$
  and routing their flows simultaneously in $H$ can result
  in a congestion of at most $|S|(|S|-1) \le |S|^2$ since each
  individual flow can be feasibly routed in $H$. This shows that
  $\sigma \le |S|^2$.
\end{proof}

\begin{remark}
 \label{rem:sparsifier}
  The proof of the preceding theorem shows that the congestion
  parameter $\rho$ can be chosen to be an additive $1$ if $G$ is a
  capacitated graph.
\end{remark}

\subsection{Routings through a small set of nodes}

We use as a black box the following result from Section 3.1\ in
\cite{CKS-sqrtn}.
\begin{prop}
\label{prop:oldsinglenode}
Let $G,H$ be a \medp instance and let $\bx$ be a fractional solution
such that there is a node $v$ that is contained in every flow path
with positive flow. Then there is a polynomial time algorithm
that routes at least $\frac{1}{12}\sum_i x_i$ pairs
from $H$ on edge-disjoint paths.
\end{prop}

\begin{remark}
  The bound of $1/12$ in the preceding proposition is not explicitly
  stated in \cite{CKS-sqrtn} but can be inferred from the arguments.
\end{remark}


Now suppose that instead of a single node $v$, there is a subset $S$
that intersects every flow path in a fractional solution $\bx$. It is
then easy to see that there is a node $v$ that intersects flow paths
of total value at least $\sum_i x_i/|S|$. We can then apply the preceding
proposition to claim that we can route $\frac{1}{12|S|}\sum_i x_i$ pairs.
We combine this with a simple re-routing argument that is relevant
to our algorithm to obtain the following.

\begin{prop}
\label{prop:withreroute}
Let $G,H$ be a matching instance of \medp and let $\bx$ be a feasible
fractional solution for it.  Suppose that there is also a second flow
that routes at least $x_i/\alpha$ flow from each terminal to some $S
\subseteq V$ where $\alpha \ge 1$. Then there is an integral routing
of at least $\frac{1}{36\alpha |S|}\sum_i x_i$ pairs.
\end{prop}

\begin{proof}
  Let $v \in S$ be the terminal which receives the most flow. Clearly
  this is of value at least $\frac{\sum_i x_i}{\alpha |S|}$. Consider
  a pair $s_it_i$ such that one of the end points, say $s_i$ sends
  $y_i \le x_i/\alpha$ flow to $v$. The other end point $t_i$ may send
  less than $y_i$ (or no flow) to $v$. We may then create a $y_i$ flow
  from $t_i$ to $v$ by using the $x_i$-flow between $s_i,t_i$ and the
  flow from $s_i$ to $v$. It is easy to see that overlaying all of the
  flows will cause capacities to be violated by a factor of at most
  $3$; we scale down the flows by a factor of $3$ to satisfy the
  capacity constraints. Via this process we can find a new fractional
  solution $\bx'$ such that (1) all the flow paths contain $v$ and (2)
  $\sum_i x'_i \ge \frac{\sum_i y_i}{3} \ge \frac{\sum_i x_i}{3\alpha |S|}$.
  The result now follows by applying Proposition~\ref{prop:oldsinglenode}.
\end{proof}

\fi

\section{\medp in $k$-sums over a family $\cG$}
\label{sec:ksums}
The goal of this section is to prove
Theorem~\ref{thm:ksums}. Throughout, we assume $\cG$ is a
minor closed family, and we wish to prove bounds for the family
$\cG_k$ obtained by $k$-sums.  In particular, we assume that
every subgraph on $k$ nodes is included in $\cG$.

Let $\cA$ be an algorithm/oracle that has the following
property: given a \medp instance on a graph $G \in \cG$ it
integrally routes $h$ pairs with congestion $\beta$ where $h$ is at
least a $1/\alpha$ fraction of the value of an optimum fractional
solution to that instance. We call $\mathcal{A}$ an {\em $(\alpha,\beta)$-oracle}.
We describe an algorithm using $\mathcal{A}$ to approximate \medp{} on
$\cG_k$. The proof is via induction on the width of a
decomposition and the number of nodes. One basic step is to take a
sparse cut $S$, lose all the flow crossing that cut, and recurse on
both sides. We need to make our recursion on treewidth on the side $S$
to which we charge the flow lost by cutting the graph. On the other
side, $V \setminus S$, we  simply recurse on the number of
nodes. The main difficulty is to show how the treewidth is
decreased on the $S$ side. The trick is a trade-off between the main
treewidth parameter and some connectivity properties in parts of the
graph with higher treewidth. To drive this we need a more refined
induction hypothesis rather than basing it only on the width of a tree
decomposition.

Given $k \leq p$, a \emph{$p$-degenerate $k$-tree decomposition} over
$\cG$ of a connected graph $G$ is a rooted tree decomposition
$\cT$ where some leaves (nodes of degree $1$) of the tree are
labelled {\em degenerate} and:
\begin{itemize}
\item for every node $X$ of $\cT$, either $G[[X]] \in \cG$ or $X$ is a
  degenerate leaf (in which case $G[[X]]$ may be arbitrary),
\item the separator corresponding to any edge $uv \in \cT$ is  of size at most $k$,
unless it is incident to a degenerate leaf, in which case it may be up to size $p$.
\end{itemize}

A pendant leaf is not necessarily degenerate but if it is not, then it
corresponds to a graph in $\cG$.  We use $(k,p)$-tree decomposition as
a shorthand notation. We call a multiflow in such a graph $G$ {\em
  flush} with the decomposition if every flow path that terminates at
some node in a degenerate leaf $L$, also intersects $S_L$ (we recall
that $S_L$ is the separator $V(L) \cap V(X)$ where $X$ is the parent
node of $L$).


We may think of graphs with such flush / degenerate decompositions as
having an {\em effective treewidth} of size $k$. In effect, we can
ignore that some separators can be larger, because we have the
separate property of flushness which we can leverage via results such
as Proposition~\ref{prop:oldsinglenode}.

The next theorem describes the algorithm for converting an LP solution
on some graph in $\cG_p$, into an integral routing. As we process this
solution, the tree decomposition becomes degenerate with the value of
$p$ fixed, $k$ is gradually reduced to $0$. We will assume that $p >
0$ for if $k=p=0$ all the separators are empty and the given connected
graph $G$ is in $\cG$ and we can simply use $\mathcal{A}$.

\begin{theorem}
  \label{thm:ksum-main}
  Let $\mathcal{A}$ be an $(\alpha,\beta)$ oracle for \medp in
  a minor-closed family $\cG$. Let $G$ be graph with a $(k,p)$-tree decomposition $\cT$ and
  suppose that $G,H$ is an instance of \medp with a fractional
  solution $\bx$ that is flush with $\cT$.  Then there is an algorithm
  with oracle $\mathcal{A}$, which computes an integral multicommodity
  flow with congestion $\beta+3$ and value $\gamma = \frac{\sum_i
    x_i}{216 \cdot \alpha p^2 3^{k}}$.  Moreover, this algorithm
  can be used to obtain the following:

\begin{enumerate}
\item an $LP$-based approximation algorithm with ratio
  $O(\alpha p^2 3^p)$ and congestion $\beta+3$ for
  \medp{} in $\cG_{p}$
\item an algorithm with approximation ratio $O((p+1) 3^p)$ and congestion
  $1$ for the class of graphs of treewidth $p$.
\end{enumerate}
\end{theorem}

\iffull
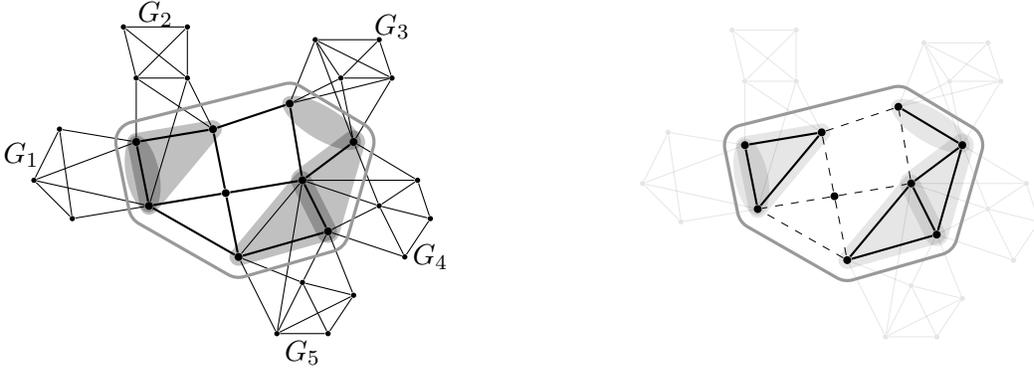
\begin{figure}
\label{fig:base-case}
\begin{tabular}{cc}

\begin{minipage}{0.48\textwidth}

\begin{tikzpicture}[x=0.17cm,y=0.17cm,>=latex]

\vertex{black} (nx32y15) at (32,15) {};
\vertex{black} (nx34y19) at (34,19) {};
\vertex{black} (nx36y15) at (36,15) {};
\vertex{black} (nx38y18) at (38,18) {};
\vertex{black} (nx42y21) at (42,21) {};
\vertex{black} (nx44y24) at (44,24) {};
\vertex{black} (nx43y27) at (43,27) {};
\vertex{black} (nx40y25) at (40,25) {};
\vertex{black} (nx41y35) at (41,35) {};
\vertex{black} (nx40y38) at (40,38) {};
\vertex{black} (nx35y38) at (35,38) {};
\vertex{black} (nx37y35) at (37,35) {};
\vertex{black} (nx13y27) at (13,27) {};
\vertex{black} (nx16y24) at (16,24) {};
\vertex{black} (nx15y31) at (15,31) {};
\vertex{black} (nx25y39) at (25,39) {};
\vertex{black} (nx25y35) at (25,35) {};
\vertex{black} (nx21y35) at (21,35) {};
\vertex{black} (nx20y39) at (20,39) {};

\Vertex{black} (nx34y27) at (34,27) {};
\Vertex{black} (nx28y26) at (28,26) {};
\Vertex{black} (nx29y21) at (29,21) {};
\Vertex{black} (nx36y23) at (36,23) {};
\Vertex{black} (nx38y30) at (38,30) {};
\Vertex{black} (nx33y33) at (33,33) {};
\Vertex{black} (nx27y31) at (27,31) {};
\Vertex{black} (nx21y30) at (21,30) {};
\Vertex{black} (nx22y25) at (22,25) {};

\fill[black,nearly transparent]
      (29.19,20.33) arc[start angle = 286, end angle = 140, radius = 0.70]
   -- (33.46,27.45) arc[start angle = 140, end angle =  27, radius = 0.70]
   -- (36.63,23.31) arc[start angle =  27, end angle = -74, radius = 0.70]
   -- cycle;
\draw[black] (nx34y27) -- (nx29y21);
\draw[black] (nx34y27) -- (nx32y15);
\draw[black] (nx32y15) -- (nx38y18);
\draw[black] (nx36y15) -- (nx34y19);
\draw[black] (nx36y15) -- (nx38y18);
\draw[black] (nx32y15) -- (nx36y15);
\draw[black] (nx32y15) -- (nx29y21);
\draw[black] (nx34y19) -- (nx32y15);
\draw[black] (nx38y18) -- (nx34y19);
\draw[black] (nx36y23) -- (nx38y18);
\draw[black] (nx34y19) -- (nx36y23);
\draw[black] (nx29y21) -- (nx34y19);

\fill[black,nearly transparent] 
      (37.46,30.72) arc[start angle = 127, end angle = -16, radius = 0.90]
   -- (36.87,22.75) arc[start angle = 344, end angle = 207, radius = 0.90]
   -- (33.20,26.60) arc[start angle = 207, end angle = 127, radius = 0.90]
   -- cycle;
\draw[black] (nx40y25) -- (nx34y27);
\draw[black] (nx43y27) -- (nx34y27);
\draw[black] (nx40y25) -- (nx44y24);
\draw[black] (nx40y25) -- (nx42y21);
\draw[black] (nx43y27) -- (nx40y25);
\draw[black] (nx44y24) -- (nx43y27);
\draw[black] (nx42y21) -- (nx44y24);
\draw[black] (nx36y23) -- (nx42y21);
\draw[black] (nx40y25) -- (nx36y23);
\draw[black] (nx38y30) -- (nx40y25);
\draw[black] (nx38y30) -- (nx43y27);

\fill[black,nearly transparent] (35.5, 31.5) circle[x radius = 3.5, y radius = 1.3, rotate = -31] {};
\draw[black] (nx41y35) -- (nx40y38);
\draw[black] (nx37y35) -- (nx41y35);
\draw[black] (nx35y38) -- (nx37y35);
\draw[black] (nx40y38) -- (nx35y38);
\draw[black] (nx37y35) -- (nx40y38);
\draw[black] (nx37y35) -- (nx38y30);
\draw[black] (nx33y33) -- (nx37y35);
\draw[black] (nx35y38) -- (nx41y35);
\draw[black] (nx38y30) -- (nx35y38);
\draw[black] (nx41y35) -- (nx33y33);
\draw[black] (nx38y30) -- (nx41y35);
\draw[black] (nx33y33) -- (nx35y38);

\fill[black,nearly transparent]
      (22.54,24.55) arc[start angle = 320, end angle = 191, radius = 0.70]
   -- (20.31,29.86) arc[start angle = 191, end angle =  99, radius = 0.70]
   -- (26.88,31.69) arc[start angle =  99, end angle = -40, radius = 0.70]
   -- cycle;
\draw[black] (nx25y35) -- (nx22y25);
\draw[black] (nx21y35) -- (nx20y39);
\draw[black] (nx25y35) -- (nx25y39);
\draw[black] (nx21y35) -- (nx25y35);
\draw[black] (nx25y39) -- (nx21y35);
\draw[black] (nx20y39) -- (nx25y39);
\draw[black] (nx25y35) -- (nx20y39);
\draw[black] (nx27y31) -- (nx25y35);
\draw[black] (nx21y35) -- (nx27y31);
\draw[black] (nx21y30) -- (nx21y35);

\fill[black,nearly transparent] (21.5, 27.5) circle[x radius = 3.05, y radius = 1.3, rotate = -78.7] {};
\draw[black] (nx13y27) -- (nx15y31);
\draw[black] (nx16y24) -- (nx13y27);
\draw[black] (nx15y31) -- (nx16y24);
\draw[black] (nx21y30) -- (nx15y31);
\draw[black] (nx22y25) -- (nx16y24);
\draw[black] (nx13y27) -- (nx22y25);
\draw[black] (nx21y30) -- (nx13y27);

\draw[very thick, black!40]
      (21.21,23.61) arc[start angle = 240, end angle = 191, radius = 1.60]
   -- (19.43,29.69) arc[start angle = 191, end angle = 104, radius = 1.60]
   -- (32.61,34.55) arc[start angle = 104, end angle =  59, radius = 1.60]
   -- (38.82,31.37) arc[start angle =  59, end angle = -16, radius = 1.60]
   -- (37.54,22.56) arc[start angle = -16, end angle = -74, radius = 1.60]
   -- (29.44,19.46) arc[start angle = 286, end angle = 240, radius = 1.60]
   -- cycle
;

\draw[black,thick] (nx36y23) -- (nx34y27);
\draw[black,thick] (nx29y21) -- (nx36y23);
\draw[black,thick] (nx29y21) -- (nx22y25);
\draw[black,thick] (nx28y26) -- (nx29y21);
\draw[black,thick] (nx22y25) -- (nx21y30);
\draw[black,thick] (nx28y26) -- (nx22y25);
\draw[black,thick] (nx28y26) -- (nx34y27);
\draw[black,thick] (nx27y31) -- (nx28y26);
\draw[black,thick] (nx34y27) -- (nx33y33);
\draw[black,thick] (nx38y30) -- (nx34y27);
\draw[black,thick] (nx27y31) -- (nx33y33);
\draw[black,thick] (nx21y30) -- (nx27y31);

\path (12,29) node {$G_1$};
\path (22.5,40) node {$G_2$};
\path (41,39) node {$G_3$};
\path (44,21) node {$G_4$};
\path (34,13.5) node {$G_5$};

\end{tikzpicture}

\end{minipage} &

\begin{minipage}{0.48\textwidth}

\begin{tikzpicture}[x=0.17cm,y=0.17cm,>=latex]

\vertex{black!10} (nx32y15) at (32,15) {};
\vertex{black!10} (nx34y19) at (34,19) {};
\vertex{black!10} (nx36y15) at (36,15) {};
\vertex{black!10} (nx38y18) at (38,18) {};
\vertex{black!10} (nx42y21) at (42,21) {};
\vertex{black!10} (nx44y24) at (44,24) {};
\vertex{black!10} (nx43y27) at (43,27) {};
\vertex{black!10} (nx40y25) at (40,25) {};
\vertex{black!10} (nx41y35) at (41,35) {};
\vertex{black!10} (nx40y38) at (40,38) {};
\vertex{black!10} (nx35y38) at (35,38) {};
\vertex{black!10} (nx37y35) at (37,35) {};
\vertex{black!10} (nx13y27) at (13,27) {};
\vertex{black!10} (nx16y24) at (16,24) {};
\vertex{black!10} (nx15y31) at (15,31) {};
\vertex{black!10} (nx25y39) at (25,39) {};
\vertex{black!10} (nx25y35) at (25,35) {};
\vertex{black!10} (nx21y35) at (21,35) {};
\vertex{black!10} (nx20y39) at (20,39) {};

\Vertex{black} (nx34y27) at (34,27) {};
\Vertex{black} (nx28y26) at (28,26) {};
\Vertex{black} (nx29y21) at (29,21) {};
\Vertex{black} (nx36y23) at (36,23) {};
\Vertex{black} (nx38y30) at (38,30) {};
\Vertex{black} (nx33y33) at (33,33) {};
\Vertex{black} (nx27y31) at (27,31) {};
\Vertex{black} (nx21y30) at (21,30) {};
\Vertex{black} (nx22y25) at (22,25) {};

\fill[black,very nearly transparent]
      (29.19,20.33) arc[start angle = 286, end angle = 140, radius = 0.70]
   -- (33.46,27.45) arc[start angle = 140, end angle =  27, radius = 0.70]
   -- (36.63,23.31) arc[start angle =  27, end angle = -74, radius = 0.70]
   -- cycle;
\draw[black,very nearly transparent] (nx34y27) -- (nx29y21);
\draw[black,very nearly transparent] (nx34y27) -- (nx32y15);
\draw[black,very nearly transparent] (nx32y15) -- (nx38y18);
\draw[black,very nearly transparent] (nx36y15) -- (nx34y19);
\draw[black,very nearly transparent] (nx36y15) -- (nx38y18);
\draw[black,very nearly transparent] (nx32y15) -- (nx36y15);
\draw[black,very nearly transparent] (nx32y15) -- (nx29y21);
\draw[black,very nearly transparent] (nx34y19) -- (nx32y15);
\draw[black,very nearly transparent] (nx38y18) -- (nx34y19);
\draw[black,very nearly transparent] (nx36y23) -- (nx38y18);
\draw[black,very nearly transparent] (nx34y19) -- (nx36y23);
\draw[black,very nearly transparent] (nx29y21) -- (nx34y19);

\fill[black,very nearly transparent] 
      (37.46,30.72) arc[start angle = 127, end angle = -16, radius = 0.90]
   -- (36.87,22.75) arc[start angle = 344, end angle = 207, radius = 0.90]
   -- (33.20,26.60) arc[start angle = 207, end angle = 127, radius = 0.90]
   -- cycle;
\draw[black,very nearly transparent] (nx40y25) -- (nx34y27);
\draw[black,very nearly transparent] (nx43y27) -- (nx34y27);
\draw[black,very nearly transparent] (nx40y25) -- (nx44y24);
\draw[black,very nearly transparent] (nx40y25) -- (nx42y21);
\draw[black,very nearly transparent] (nx43y27) -- (nx40y25);
\draw[black,very nearly transparent] (nx44y24) -- (nx43y27);
\draw[black,very nearly transparent] (nx42y21) -- (nx44y24);
\draw[black,very nearly transparent] (nx36y23) -- (nx42y21);
\draw[black,very nearly transparent] (nx40y25) -- (nx36y23);
\draw[black,very nearly transparent] (nx38y30) -- (nx40y25);
\draw[black,very nearly transparent] (nx38y30) -- (nx43y27);

\fill[black,very nearly transparent] (35.5, 31.5) circle[x radius = 3.5, y radius = 1.3, rotate = -31] {};
\draw[black,very nearly transparent] (nx41y35) -- (nx40y38);
\draw[black,very nearly transparent] (nx37y35) -- (nx41y35);
\draw[black,very nearly transparent] (nx35y38) -- (nx37y35);
\draw[black,very nearly transparent] (nx40y38) -- (nx35y38);
\draw[black,very nearly transparent] (nx37y35) -- (nx40y38);
\draw[black,very nearly transparent] (nx37y35) -- (nx38y30);
\draw[black,very nearly transparent] (nx33y33) -- (nx37y35);
\draw[black,very nearly transparent] (nx35y38) -- (nx41y35);
\draw[black,very nearly transparent] (nx38y30) -- (nx35y38);
\draw[black,very nearly transparent] (nx41y35) -- (nx33y33);
\draw[black,very nearly transparent] (nx38y30) -- (nx41y35);
\draw[black,very nearly transparent] (nx33y33) -- (nx35y38);

\fill[black,very nearly transparent]
      (22.54,24.55) arc[start angle = 320, end angle = 191, radius = 0.70]
   -- (20.31,29.86) arc[start angle = 191, end angle =  99, radius = 0.70]
   -- (26.88,31.69) arc[start angle =  99, end angle = -40, radius = 0.70]
   -- cycle;
\draw[black,very nearly transparent] (nx25y35) -- (nx22y25);
\draw[black,very nearly transparent] (nx21y35) -- (nx20y39);
\draw[black,very nearly transparent] (nx25y35) -- (nx25y39);
\draw[black,very nearly transparent] (nx21y35) -- (nx25y35);
\draw[black,very nearly transparent] (nx25y39) -- (nx21y35);
\draw[black,very nearly transparent] (nx20y39) -- (nx25y39);
\draw[black,very nearly transparent] (nx25y35) -- (nx20y39);
\draw[black,very nearly transparent] (nx27y31) -- (nx25y35);
\draw[black,very nearly transparent] (nx21y35) -- (nx27y31);
\draw[black,very nearly transparent] (nx21y30) -- (nx21y35);

\fill[black,very nearly transparent] (21.5, 27.5) circle[x radius = 3.05, y radius = 1.3, rotate = -78.7] {};
\draw[black,very nearly transparent] (nx13y27) -- (nx15y31);
\draw[black,very nearly transparent] (nx16y24) -- (nx13y27);
\draw[black,very nearly transparent] (nx15y31) -- (nx16y24);
\draw[black,very nearly transparent] (nx21y30) -- (nx15y31);
\draw[black,very nearly transparent] (nx22y25) -- (nx16y24);
\draw[black,very nearly transparent] (nx13y27) -- (nx22y25);
\draw[black,very nearly transparent] (nx21y30) -- (nx13y27);

\draw[very thick, black!40]
      (21.21,23.61) arc[start angle = 240, end angle = 191, radius = 1.60]
   -- (19.43,29.69) arc[start angle = 191, end angle = 104, radius = 1.60]
   -- (32.61,34.55) arc[start angle = 104, end angle =  59, radius = 1.60]
   -- (38.82,31.37) arc[start angle =  59, end angle = -16, radius = 1.60]
   -- (37.54,22.56) arc[start angle = -16, end angle = -74, radius = 1.60]
   -- (29.44,19.46) arc[start angle = 286, end angle = 240, radius = 1.60]
   -- cycle
;

\draw[black, dashed,very thin] (nx29y21) -- (nx22y25);
\draw[black, dashed,very thin] (nx28y26) -- (nx29y21);
\draw[black, dashed,very thin] (nx28y26) -- (nx22y25);
\draw[black, dashed,very thin] (nx28y26) -- (nx34y27);
\draw[black, dashed,very thin] (nx27y31) -- (nx28y26);
\draw[black, dashed,very thin] (nx34y27) -- (nx33y33);
\draw[black, dashed,very thin] (nx27y31) -- (nx33y33);

\draw[black,thick] (nx38y30) -- (nx34y27);
\draw[black,thick] (nx36y23) -- (nx34y27);
\draw[black,thick] (nx29y21) -- (nx34y27);
\draw[black,thick] (nx21y30) -- (nx27y31);
\draw[black,thick] (nx22y25) -- (nx21y30);
\draw[black,thick] (nx29y21) -- (nx36y23);
\draw[black,thick] (nx33y33) -- (nx38y30);
\draw[black,thick] (nx38y30) -- (nx36y23);
\draw[black,thick] (nx22y25) -- (nx27y31);

\end{tikzpicture}
\end{minipage}
\end{tabular}

\caption{Illustration of the base case with $k=0$ and $p = 3$: the graph
corresponding to each degenerate leaf is replaced by a sparsifier on the
associated separator.}
\end{figure}

\fi

The proof of the preceding theorem is somewhat long and technical and
 occupies the rest of the section.  To help the exposition we break
it up into several components.  The proof proceeds by induction on $k$
and the number of nodes.  The base case with $k=0$ is non-trivial
and we treat it first. 
\ifabstract
The reader may find Figures~\ref{fig:base-case} and \ref{fig:induction-step}
in the appendix helpful in following the proof.
\fi

\paragraph{The base case:} We can assume without loss of generality
that $G$ is connected. 
Throughout we assume a fixed $p \geq 1$ and consider a 
decomposition together with a flush
fractional routing $\bx$ as described.  Let $x: V \to \mathbb{R}$ be the
marginal values of $\bx$. Again we use $x_i$ to denote the common value
$x(s_i)=x(t_i)$. Hence $\val{\bx} = \sum_i x_i= \frac{1}{2} \sum_{v \in
  V(G)} x(v)$ where we use the notation $\val{\bx}$ to denote the value of
the flow $\bx$.

Now assume that $k=0$ and $p \ge 1$. We may assume that $\cT$ has more
than one node otherwise $G \in \cG$ and we can apply $\mathcal{A}$. 
Since $G$ is connected, each separator corresponding to
an edge of $\cT$ must be non-trivial. Since $k=0$, each edge of $\cT$
must be incident to a degenerate leaf. It follows that $\cT$ is a
single edge between two degenerate leaves or a star whose leaves are
all degenerate. If $\cT$ is a single edge between two degenerate
leaves, all flow paths intersect the separator of size $p$ associated
with the edge; hence by Proposition~\ref{prop:oldsinglenode}, there is
an integral routing of value at least $(\sum_i x_i)/(12p)$.  We
therefore restrict our attention to the case when $\cT$ is a star with
$\ell$ leaves. Let $G^* = G[X]$ where $X$ is the bag at the
center/root; we observe that $G^* \in \cG$. Let $X_i$ be the bag at
the $i$'th leaf.  We let $G_i$ denote the graph obtained from $G[X_i]$
after removing the edges between the separator nodes $S_i = X \cap
X_i$.


The base case of theorem assertion (2) on treewidth $p$ graphs holds
as follows. By flushness, all flow paths intersect $G^*$, which has at
most $p+1$ nodes. Hence there is some node which is receiving at least
$(\sum_i x_i)/(p+1)$ of this flow. Hence by
Proposition~\ref{prop:oldsinglenode}, there is a (congestion 1)
integral routing of value at least $(\sum_i x_i)/(12(p+1))$.


Now consider the general case with a minor-closed class $\cG$ and
$\cT$ a star whose center is $G^* \in \cG$.  We proceed in the
following steps:

\begin{enumerate}
\item Using the flushness property move the terminals in each $G_i$ to
  the separator $S_i$ that is contained in $G^*$ (using
  Lemma~\ref{lem:rr}).

\item In $G$ replace each $G_i$ by a $(p^2,2)$ sparsifier on $S_i$ to
  obtain a new graph $G'$. Via the sparsifier property,
  scale the flow in $G$ down by a factor of $p^2$ to obtain a corresponding
  feasible flow in $G'$.

\item Apply the algorithm $\mathcal{A}$ on the new instance in $G'$.

\item Transfer the routing in $G'$ to a routing in $G$ with
  additive $+1$ via the sparsifier property.

\item Use the second part of Lemma~\ref{lem:rr} to convert the routing
  in $G$ into a solution for our original instance before the terminals
  were moved (incur an additional $+2$ additive congestion).
\end{enumerate}

\noindent
We describe the steps in more detail. We observe that the graphs $G_i$
are edge-disjoint. The first step is a simple application of
Lemma~\ref{lem:rr} where for each $i,$ we move any terminals in
$G_i-S_i$ to the separator $S_i$ via clustering. This is
possible because of the flushness assumption; if $P$ is a flow path
with an endpoint in $G_i-S_i$ then that path intersects
$S_i$.  This  incurs a factor $5$ loss in the value of the new
flow we work with (and it  incurs an additive $2$ congestion when
we convert back to an integral solution for our original instance).
To avoid notational overload we let $f$ be the flow for the new instance
which is at least $\frac{1}{5}$th of the original flow. 

After the preceding step no node in $G_i-S_i$ is the end point of a terminal.
In step (2), we can simultaneously replace each $G_i$ by a
$(p^2,2)$-sparsifier $F_i$ on $S_i$ --- see
Theorem~\ref{thm:sparsifier}. Call the new instance $G'$ and note that
since we only added edges to the separators $S_i$, $G'$ is a subgraph
of $G[[X]]$ and hence $G' \in \cG$. At this step, we also need to
convert our flows in $G_i$'s to be flows in $G'$. The sparsifier
guarantees that any multicommodity flow on $S_i$ that is feasible in
$G_i$ can be routed in $F_i$ with congestion $p^2$. Hence, scaling the
flow down by $p^2$ guarantees its feasibility in $G'$.

We now work with the new flow $\bx'$ in the graph $G'$ and apply
$\mathcal{A}$ to obtain a routing of size $\val{\bx'}/\alpha \ge
\val{\bx}/(5p^2 \alpha)$ with congestion $\beta$. We must now convert
this integral routing to one in $G$. Again, for each $i$, there is an
embedded integral routing in $F_i$ which will be re-routed in
$G_i$. We incur an additive $1$ congestion for this; see
Remark~\ref{rem:sparsifier}. Finally we apply the second part of
Lemma~\ref{lem:rr} to route the original pairs in $G$ before they were moved
to the separators, incurring an additive congestion of $2$.

Thus the total number of pairs routed is at least $\val{\bx}/(5 \alpha p^2)$
and the overall congestion of the routing is $\beta+3$. This proves
the base case when $k=0$.

\paragraph{The induction step:}
Henceforth, we assume that $p \geq k > 0$ and that $\cT$ contains at
least one edge $e$ with the associated separator $V_e$ (the
intersection of the bags at the two end points of $e$) of size equal
to $k$; otherwise $\cT$ is a star with degenerate leaves as in the
base case, or we may use $k-1$.  We consider an easy setting when
there is a flow $g$ that simultaneously routes $x(v)/6$ amount from
each vertex $v$ to the set $V_e$. (Note that checking the existence of the
desired flow to $V_e$ can be done by a simple maximum-flow
computation.)  We then obtain an integral (congestion 1) flow of size
$(\sum_i x_i)/(216 k)$ via Proposition~\ref{prop:withreroute} which is
sufficient to establish the induction step for $k$.

Assume now that there is no such flow $g$. Then there is a cut $U
\subset V \setminus V_e$ with $c(U) := c(\delta(U)) < \frac{1}{6}
x(U)$.  We may assume that $U$ is minimal and central ($G[U]$ and $G[V
\setminus U]$ are connected). Such a cut can be recovered from
the maximum flow computation. We now work with a reduced flow $\bx'$
obtained from $\bx$ by eliminating any flow path that intersects
$\delta(U)$.  We also let $x'$ be the marginals for $\bx'$.  Obviously
we have
$$\val{\bx} - \val{\bx'} \leq c(U) < \frac{x(U)}{6}.$$

\noindent
Let $f_U,f_\bU$ be the flow vectors obtained from $\bx'$, where $f_U$
only uses the flow paths contained in $U$, and $f_\bU$ uses the flow
paths contained in $V \setminus U$. The idea is that we recurse on
$G[U]$ and $G[V\setminus U]$. We  modify the instance on $G[U]$ to
ensure that it has a $(k-1,p)$-tree decomposition and charge the
lost flow to this side. The recursion on $G[V\setminus U]$ is based on
reducing the number of nodes, the width is not reduced. Reducing the
width on the $U$ side and ensuring the flushness property is not
immediate; it requires us to modifying $f_U$ and in the process we may
lose further flow. We explain this process before analyzing the number
of pairs routed by the algorithm.

We note that $G[U]$ and $G[V \setminus U]$ easily admit $(k,p)$-tree
decompositions, by intersecting the nodes of $\cT$ with $U$ and $V
\setminus U$ respectively, and removing the empty nodes; recall $G[U]$
and $G[V\setminus U]$ are connected.  Denote these by $\cT_U$ and
$\cT_\bU$ respectively.  Degenerate leaves of $\cT_U$ and $\cT_\bU$
are the same as the degenerate leaves of $\cT$.  Some of these are
``split'' by the cut, otherwise they are simply assigned to either
$\cT_U$ or $\cT_\bU$. If split, the two ``halves'' go to appropriate
sides of the decomposition. The flows $f_U,f_\bU$ will be flush with
each such degenerate leaf (if the leaf is split, any flow path that
crosses the cut is removed).


\iffull
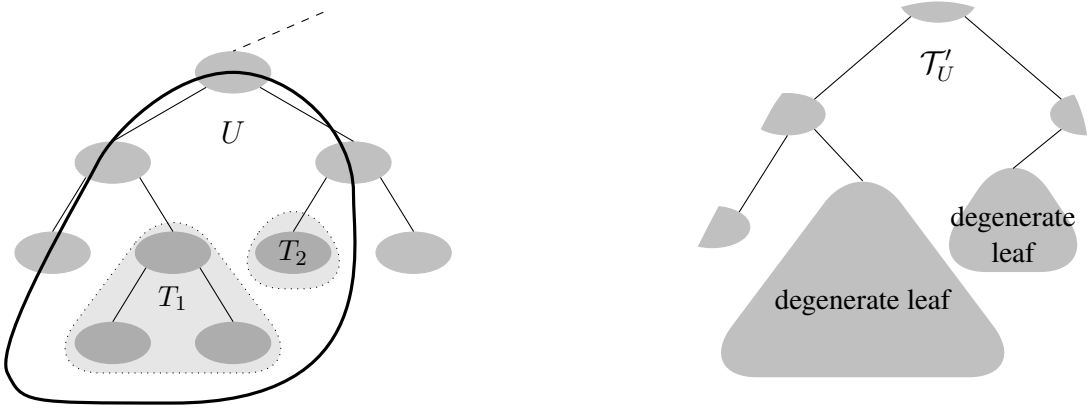
\begin{figure}
\begin{tabular}{cc}
\begin{minipage}{0.48\textwidth}

\begin{tikzpicture}[x=0.4cm,y=0.4cm,>=latex]

\nellipse{black,nearly transparent} (a) at (8,10) {};
\nellipse{black,nearly transparent} (b) at (4,7) {};
\nellipse{black,nearly transparent} (c) at (12,7) {};
\nellipse{black,nearly transparent} (d) at (2,4) {};
\nellipse{black,nearly transparent} (g) at (14,4) {};

\nellipse{black,nearly transparent} (e) at (6,4) {};
\nellipse{black,nearly transparent} (f) at (10,4) {};
\nellipse{black,nearly transparent} (h) at (4,1) {};
\nellipse{black,nearly transparent} (i) at (8,1) {};

\draw[dashed] (11,12) -- (a.north);
\draw
  (a.south west) -- (b.north)
  (a.south east) -- (c.north)
  (b.south west) -- (d.north)
  (c.south east) -- (g.north)
;

\draw
  (b.south east) -- (e.north)
  (c.south west) -- (f.north)
  (e.south west) -- (h.north)
  (e.south east) -- (i.north)
;

\draw[very thick,black] 
  (6,-1) .. controls (13,-1) and (12,4) ..
  (12,6) .. controls (12,8) and (10,10) ..
  (8,10) .. controls (6,10) and (4,8) ..
  (3.5,7) .. controls (3,6) and (0,1) ..
  (0.5,0) .. controls (1,-1) and (1,-1) .. (6,-1);

\path[rounded corners=24pt,fill=black,very nearly transparent]
  (6,6) -- (1.5,0) -- (10.5,0) -- cycle;
\path[rounded corners=24pt,draw=black,dotted]
  (6,6) -- (1.5,0) -- (10.5,0) -- cycle;
\node (t1) at (6,2.5) {$T_1$};

\path[rounded corners=24pt,fill=black,very nearly transparent]
 (10,6.4) -- (7.6,2.8) -- (12.4,2.8) -- cycle;
\path[rounded corners=24pt,draw=black,dotted]
 (10,6.4) -- (7.6,2.8) -- (12.4,2.8) -- cycle;
\node (t2) at (10,4) {$T_2$};

\node at (8,8) {{\large $U$}};

\end{tikzpicture}

\end{minipage} &

\begin{minipage}{0.48\textwidth}

\begin{tikzpicture}[x=0.5cm,y=0.5cm,>=latex]

\path[clip] 
  (6,-1) .. controls (13,-1) and (12,4) ..
  (12,6) .. controls (12,8) and (10,10) ..
  (8,10) .. controls (6,10) and (4,8) ..
  (3.5,7) .. controls (3,6) and (0,1) ..
  (0.5,0) .. controls (1,-1) and (1,-1) .. (6,-1);

\nellipse{black,nearly transparent} (a) at (8,10) {};
\nellipse{black,nearly transparent} (b) at (4,7) {};
\nellipse{black,nearly transparent} (c) at (12,7) {};
\nellipse{black,nearly transparent} (d) at (2,4) {};


\draw[dashed] (11,12) -- (a.north);
\draw
  (a.south west) -- (b.north east)
  (a.south east) -- (c.north west)
  (b.south) -- (d.north east)
  (10,6.4) ++(0,-12pt) -- (c.south west)
;

\draw
  (6,6) ++(0,-11pt)-- (b.south east);

\path[rounded corners=24pt,fill=black,nearly transparent]
  (6,6) -- (1.5,0) -- (10.5,0) -- cycle;
\node (t1) at (6,2) {degenerate leaf};

\path[rounded corners=24pt,fill=black,nearly transparent]
 (10,6.4) -- (7.6,2.8) -- (12.4,2.8) -- cycle;

\node[align=center] (t2) at (10,3.8) {degenerate\\ leaf};

\node at (8,8.3) {{\large $\mathcal{T}'_U$}};

\end{tikzpicture}
\end{minipage}
\end{tabular}

\caption{Induction step: reducing the effective treewidth of $G[U]$
by creating degenerate leaves. $T_1,T_2$ are maximal subtrees rooted at separators of size $k$, converted to degenerate leaves.}
\label{fig:induction-step}
\end{figure}
\fi

We  proceed in $\cT_\bU$ by induction on the number of
nodes. However, since we charge the lost flow to the cut (i.e., to
$\cT_U$) we modify $\cT_U$ to obtain a $(k-1,p)$-tree decomposition.
We state a lemma that accomplishes this.

\begin{lemma}
  \label{lem:ksum-recursion}
  For the residual instance on $G[U]$ with flow $f_U$ we can either
  route $\frac{1}{216k} \cdot \val{f_U}/2$ pairs integrally or find a
  $(k-1,p)$-tree decomposition $\cT'_U$ and a reduced flow
  vector $f'_U$ that is flush with $\cT'_U$ with $\val{f'_U} \ge \val{f_U}/2$.
\end{lemma}

We postpone the proof of the above lemma and proceed to finish the
recursive analysis.  We apply the induction hypothesis for $k$ on
$\cT_\bU$ with the number of nodes reduced; hence the algorithm routes
at least $\frac{\val{f_\bU}}{\alpha p^2 216 \cdot 3^{k}}$
pairs in $G[V\setminus U]$ with congestion at most $\beta+3$.  For
$G[U]$ we consider two cases based on the preceding lemma.  In the
first case the algorithm directly routes $\frac{1}{216k} \cdot
\val{f_U}/2$ pairs integrally in $G[U]$. In the second case we recurse on
$G[U]$ with the flow $f'_U$ that is flush with respect to the
$(k-1,p)$-tree decomposition $\cT'_U$; by the induction
hypothesis the algorithm routes at least
$\frac{\val{f'_U}}{\alpha p^2 216 \cdot 3^{k-1}}$ pairs with
congestion $\beta+3$. Since the number of pairs routed in this second
case is less than in the first case, we may focus on it, as we now show  that the
total number of pairs routed in $G[U]$ and $G[V\setminus U]$ satisfies
the induction hypothesis for $k$. We first observe that $\val{f} \le \val{f_U}
+ \val{f_\bU} + c(U)$. 
Moreover, $c(U) < x(U)/6$, and since at least $x(U)/2$ flow
originated in $U$, and we lost at most $c(U)$ of this flow, we have
$\val{f_U} \geq 2 c(U)$.  The total number of pairs routed is at least
\begin{align*}
\frac{\val{f'_U}}{216 \alpha p^2 \cdot 3^{k-1}} + \frac{\val{f_\bU}}{216 \alpha p^2  \cdot 3^{k}}
  &\geq \frac{3\val{f'_U}+\val{f_\bU}}{216 \alpha p^2  \cdot 3^{k}} \\
  &\geq \frac{(3/2)\val{f_U}+\val{f_\bU}}{216 \alpha p^2  \cdot 3^{k}} \\
  &\geq \frac{\val{f_U}+\val{f_\bU}+c(U)}{216 \alpha p^2  \cdot 3^{k}} \\
  &\geq \frac{\val{\bx}}{216 \alpha p^2  \cdot 3^{k}} ,
\end{align*}
which establishes the induction step for $k$. A very similar analysis
shows the slightly stronger bound for treewidth $p$ graphs --- we simply
have to use the stronger induction hypothesis in the preceding calculations
and observe that the base case analysis also proves the desired stronger
hypothesis.

\paragraph{Proof of Lemma~\ref{lem:ksum-recursion}:}
Recall that $\cT_U$ is a $(k,p)$-degenerate decomposition for $G[U]$
and that $f_U$ is flush with respect to $\cT_U$. However, we wish to
find a $(k-1,p)$-degenerate decomposition. Recall that $U$ is disjoint
from $V_e$, the separator of size $k$ associated with an edge $e$ of
$\cT$.  We can assume that $\cT$ is rooted and without loss of
generality that $\cT_U$ is a sub-tree of $\cT_e$. The reason that
$\cT_U$ may not be a $(k-1,p)$-tree decomposition is that it may
contain a edge $e'$ not incident to a degenerate leaf such that
$|V_{e'}| = k$. $V_{e'}$ was then also a separator of $G$ since $\cT$
is a $(k,p)$-tree decomposition. By centrality of $U$, every node that
$V_{e'}$ separates from $V_e$ must then be in $U$; therefore, letting
$U' = V(\cT_{e'})$ (the union of all bags contained in $\cT_{e'}$), we
have $U' \subset U$.  We claim that we can route $x(v)/6$ from each
$v \in U'$ to $V_{e'}$ in $G'=G_{e'}$ (this is the graph induced by
$G[U']$ but edges between the separator nodes $V_{e'}$ removed); this
follows from the minimality of $U$ since a cutset induced by $W
\subset U' \setminus V_{e'}$ is also a cutset of $G$.  We define a
$(k-1,p)$-tree decomposition $\cT'_U$, by contracting every {\em
  maximal} subtree of $\cT_U$ rooted at such a separator of size
$k$. Each such subtree identifies a new degenerate leaf in
$\cT'_U$. However, the flow $f_U$ may not be flush with $\cT'_U$ due
to the creation of new degenerate leaves. We try to amend this by
dropping flow paths with end points in a new degenerate leaf $L$ that does
not intersect the separator $S_L$. Let $f'_U$ be the residual flow;
observe that by definition $f'_U$ is flush with respect to
$\cT'_U$. Two cases arise.

\begin{itemize}
\item If $\val{f'_U} \ge \val{f_U}/2$ then we have the desired
  degenerate $(k-1,p)$-decomposition of $G[U]$.
\item Else at least $\val{f_U}/2$ of the flow is being
  routed completely within the new degenerate leaves $L$.  Moreover,
  these graphs and flows are edge-disjoint.  Note also 
  in such an $L$, we have that the terminals involved can simultaneously
  route $x(v)/6$ each to $S_L$. We can thus obtain a constant fraction
  of the profit by applying the method from
  Proposition~\ref{prop:withreroute}, separately to every new leaf. In
  particular, we route at least $\frac{1}{216k} \cdot \val{f_U}/2$ of the
  pairs. In this case, we no longer need to recurse on $\cT_U$.
\end{itemize}
This finishes the proof of the lemma.

Finally, the main inductive claim implies the claimed algorithmic
results since we can start with a proper $p$-decomposition $\cT$ of $G
\in \cG_p$ (viewed as a $(p,p)$-tree decomposition with no degenerate
leaves) and an arbitrary multiflow on its support (since there are no
degenerate leaves the flushness is satisfied) in order to begin our
induction. This finishes the proof. \qed

\iffull
\section{MEDP in Bounded Genus Graphs}
\label{sec:genus}
In this section we consider \medp and obtain an approximation ratio
that depends on the genus of the given graph. Throughout we use $g$ to
denote the genus of the given graph $G$; we   assume that
$g > 0$ since we already understand planar graphs. We assume that
the given instance of \medp is a matching instance in which
the terminals are degree $1$ leaves and that each terminal participates
in exactly one pair. An instance of \medp with this restriction is
characterized by a tuple $(G,X,M)$ where $G$ is the graph and $X$ is
the set of terminals and $M$ is a matching on the terminals
corresponding to the given pairs.  We  also assume that the degree
of each node is at most $4$; this can be arranged without changing the
genus by replacing a high-degree node by a grid (see \cite{Chekuri04b} for the
description for planar graphs which generalizes easily for any
surface).  Let $\bx$ be a feasible fractional solution to the
multicommodity flow based linear programming relaxation.  We use again the
terminology $\val{\bx}$ to denote $\sum_i x_i$, the fractional
amount of flow routed by $\bx$.  We let $\beta_g$ denote the flow-cut
gap for product-multicommodity flow instances in a graph of genus $g$;
this is known to be $O(\log (g+1))$ \cite{lee2010genus}. We  implicitly
assume that $\beta_g$ is an effective upper bound on the flow-cut gap
in that there is a polynomial-time algorithm that outputs a sparse cut
no worse than $\beta_g$ times the maximum concurrent flow for
a given product multicommodity instance on a genus $g$ graph.
The main result is the following.

\begin{theorem}
  \label{thm:genus_main}
  Let $\bx$ be a feasible fractional solution to a \medp instance in a
  graph of genus $g > 0$. Then
  $\Omega(\val{\bx}/\gamma_g)$ pairs can be routed with congestion
  $3$ where $\gamma_g = O(g \log^2 (g+1))$.
\end{theorem}

As we remarked previously we do not have currently have a polynomial-time
algorithm for the guarantee that we establishes in the preceding theorem.
There are three high-level ingredients in establishing the preceding theorem.
\begin{itemize}
\item A constant factor approximation with congestion $2$ for planar
  graphs based on the LP relaxation \cite{CKS-planar-constant,SeguinS11}.
\item An $O(g)$-approximation with congestion $3$ for graphs with
  genus $g$ when the terminals are well-linked. This extends the
  result in \cite{CKS-well-linked} for planar graphs to
  bounded genus graphs via the use of grid minors.
\item A modification of the well-linked decomposition of
  \cite{CKS-well-linked} that terminates the decomposition
   when the graph is planar even if the terminals are not well-linked.
\end{itemize}

We first give some formal definitions on well-linked sets;
the material follows \cite{CKS-well-linked} closely.

\medskip
\noindent {\bf Well-linked Sets:} Let $X \subseteq V$ be a set of
nodes and let $\pi : X \rightarrow [0,1]$ be a weight function on $X$.
We say $X$ is {\em $\pi$-flow-well-linked} in $G$ if there is a
feasible multicommodity flow in $G$ for the following demand matrix:
between every unordered pair of terminals $u,v \in X$ there is a
demand $\pi(u)\pi(v)/\pi(X)$ (other node pairs have zero demand).  We
say that $X$ is {\em $\pi$-cut-well-linked} in $G$ if $|\delta(S)| \ge
\pi(S \cap X)$ for all $S$ such that $\pi(S \cap X) \le \pi(X)/2$.  It
can be checked easily that if a set $X$ is $\pi$-flow-linked in $G$,
then it is $\pi/2$-cut-linked.  If $\pi(u) = \alpha$ for all $u \in
X$, we say that $X$ is {\em $\alpha$-flow(or cut)-well-linked}.  If
$\alpha=1$ we simply say that $X$ is well-linked.  Given $\pi:X
\rightarrow [0,1]$ one can check in polynomial time whether $X$ is
$\pi$-flow-well-linked or not via linear programming.

One can efficiently find an approximate sparse cut
if $X$ is not $\pi$-flow-well-linked via the algorithmic aspects
of the product flow-cut gap; the lemma below follows from
\cite{lee2010genus}.

\begin{lemma}
\label{lem:sparsecut}
Let $G=(V,E)$ be  a graph of genus at most $g > 1$.
Let $X \subseteq V$ and $\pi : X \rightarrow [0,1]$. There is
a polynomial-time algorithm that given $G$, $X$ and $\pi$ decides
whether $X$ is $\pi$-flow-well-linked in $G$ and if not outputs
a set $S \subseteq V$ such that $\pi(S) \le \pi(V\setminus S)$ and
$|\delta(S)| \le \beta_g \cdot \pi(S)$ where $\beta_g = O(\log (g+1))$.
\end{lemma}


We now formally state the theorems that correspond to the high-level
ingredients. The first is a constant factor approximation for
routing in planar graphs.

\begin{theorem}[\cite{CKS-planar-constant,SeguinS11}]
  \label{thm:planar-constant}
  Let $\bx$ be a feasible fractional solution to a \medp instance in
  a planar graph. Then there is a polynomial-time algorithm that routes
  $\Omega(\val{\bx})$ pairs with congestion $2$.
\end{theorem}

The second ingredient is an $O(g)$-approximation if the terminals
are well-linked. More precisely we have the following theorem which
we prove in Section~\ref{subsec:well-linked-genus}.
\begin{theorem}
  \label{thm:well-linked-genus}
  Let $\bx$ be a feasible fractional solution to a \medp instance in a
  graph of genus $g > 0$. Moreover, suppose the terminal set $X$ is
  $\pi$-flow-well-linked in $G$ where $\pi(v) = \rho \cdot x(v)$ for some
  scalar $\rho \le 1$. Then there is an algorithm
  that routes $\Omega(\rho \cdot \val{\bx}/g)$ pairs with congestion
  $3$.
\end{theorem}

The last ingredient is the following adaptation of the well-linked
decomposition from \cite{CKS-well-linked}.
\begin{theorem}
  \label{thm:genus-decomp}
  Let $(G,X,M)$ be an instance of \medp in a graph of genus at most
  $g$ and let $\bx$ be a feasible fractional solution. Then there is a
  polynomial-time algorithm that decomposes the given instances
  into several instances $(G_1,X_1,M_1),\ldots,(G_h,X_h,M_h)$
  with the following properties.
  \begin{itemize}
  \item The graphs $G_1,G_2,\ldots,G_h$ are node-disjoint subgraphs of $G$.
  \item For $1 \le j \le h$, $M_j \subseteq M$ and the end points of
    $M_j$ are in $G_j$.
  \item For $1 \le j \le h$, there is a feasible fractional solution
    $\bx^j$ for the instance $(G_j,X_j,M_j)$ such that $\sum_j \val{\bx^j}
    = \Omega(\val{\bx})$.
  \item For $1 \le j \le h$, $G_i$ is planar or the terminals $X_j$
    are $\bx^j/(10\beta_g \log (g+1))$-flow-well-linked  in $G_j$.
  \end{itemize}
\end{theorem}

Assuming the preceding three theorems we finish the proof of
Theorem~\ref{thm:genus_main}. Let $G$ be a graph of genus $\le g$.  We
apply the decomposition given by Theorem~\ref{thm:well-linked-genus}
which reduces the original problem instance $(G,X,M)$ to a collection
of separate instances $(G_1,X_1,M_1),\ldots,(G_h,X_h,M_h)$.  Note that
pairs in these new instances are from the original instance and
routings in these instances can be combined into a routing in the
original graph $G$ since the $G_1,G_2,\ldots,G_h$ are node-disjoint
(and hence also edge-disjoint). The total fractional solution value in
the new instances is $\Omega(\val{\bx}/\beta_g \log (g+1))$. If $G_j$
is planar we use Theorem~\ref{thm:planar-constant} to route
$\Omega(\val{\bx^j})$ pairs from $M_j$ in $G_j$ with congestion
$2$. If $G_j$ is not planar, then $X_j$ is $\bx^j/(10\beta_g \log
(g+1))$-flow-well-linked. Then, via
Theorem~\ref{thm:well-linked-genus}, we can route $\Omega(\val{\bx^j}/(g
\beta_g \log (g+1)))$ pairs from $M_j$ in $G_j$ with congestion $4$. Thus
we route in total $\Omega(\sum_j \val{\bx_j}/(g \beta_g \log (g+1))$ pairs.
Since $\sum_j \val{\bx_j} = \Omega(\val{x})$, we route a total of
$\Omega(\val{\bx}/(g \beta_g \log (g+1))) = \Omega(\val{\bx}/(g \log^2 (g+1)))$
pairs from $M$ in $G$ with congestion $3$.

\subsection{Proof of Theorem~\ref{thm:genus-decomp}}
\label{subsec:genus-decomp}

We start with a well-known fact.

\begin{prop}
  \label{prop:genus}
  Let $G=(V,E)$ be a connected graph of genus $g$. Let $S \subset V$
  be such that $G_1 = G[S]$ and $G_2=G[V\setminus S]$ are both connected.
  Then $g_1+g_2 \le g$ where $g_1$ is the genus of $G$ and $g_2$ is the
  genus of $G_2$.
\end{prop}

The algorithm below is an adaptation of the well-linked decomposition
algorithm from \cite{CKS-well-linked} that recursively partitions a
graph if the terminals are not well-linked (with a certain
parameter). In the adaptation below we stop the partitioning if the
graph becomes planar even if the terminals are not well-linked. We
start with an instance $(G,X,M)$ and an associated fractional solution
$\bx$. We fix a particular selection of flow paths for the solution $\bx$.
The algorithm recursively cuts $G$ into subgraphs by removing edges.
The flow paths that use the removed edges are lost and each connected
component retains flow corresponding to those flow paths which are
completely contained in that component.

\smallskip
\noindent
{\bf Decomposition Algorithm:}

\begin{enumerate}
\item \label{alg:term1} If $\val{\bx} < 10 \beta_g \log (g+1)$ or if $G$ is
  a planar graph stop and output $(G,X,M)$ with fractional solution
  $\bx$.

\item Else if $X$ is $\bx/(10 \beta_g \log (g+1))$ flow-well-linked in $G$ then
  stop and output $(G,X,M)$ with fractional solution $\bx$.

\item Else find a sparse cut $\delta(S)$ such that $x(S) \le
  \val{\bx}/2$ and $|\delta(S)| \le 1/(10 \log (g+1)) \cdot x(S)$. Let
  $G_1=G[S]$ and $G_2=G[V\setminus S]$. (Assume wlog that $G_1,G_2$
  are connected).  Let $(G_1,X_1,M_1) $ and $(G_2,X_2,M_2)$ be the
  induced instances on $G_1$ and $G_2$ with fractional solutions $\bx^1$
  and $\bx^2$ respectively.  Recurse
  separately on $(G_1,X_1,M_1) $ and $(G_2,X_2,M_2)$.
\end{enumerate}

We now prove that the above algorithm outputs a decomposition as
stated in Theorem~\ref{thm:genus-decomp}.  The only property that is
non-trivial to see is the one about the total flow retained in the
decomposition. As in \cite{CKS-well-linked} it is easier to upper
bound the total number of edges cut by the algorithm which in turn
upper bounds the total amount of flow from the original solution $\bx$
that is lost during the decomposition. We write a recurrence for this
as follows.  Let $L(a,r)$ be total number of edges cut by the
algorithm if $a$ is the amount of flow in the graph and $r\le g$ is
the genus. The base case is $L(a,0) = 0$ since the algorithm stops
when the graph is planar. The algorithm cuts and recurses only when
the terminals are not $\bx/(10 \beta_g \log (g+1))$ flow-well linked.
Suppose $H$ is the current graph with flow value $a = \sum_v x(v)/2$
and genus $r$, and $H$ is partitioned into $H_1$ and $H_2$. Let $a_1$
and $a_2$ be the total flow values in $H_1$ and $H_2$ and let
$r_1$ and $r_2$ be their genus respectively. We have $a_1+a_2 \le a$ and
by Proposition~\ref{prop:genus} we have $r_1+r_2 \le r$.  We claim
that the number of edges cut is at most $\frac{1}{4 \log (g+1)}
\min\{a_1,a_2\}$. To see this, suppose $\delta(S)$ is the sparse cut
in $H$ that resulted in $H_1$ and $H_2$ with $H_1=H[S]$ and $H_2 =
H[V\setminus S]$ and $a_1 = \min\{a_1,a_2\}$. Let $\bx$ be the
fractional solution in $H$ and $\bx'$ the solution after the
partitioning. We have $a = x(V(H))/2$. Since the terminals are not
$\bx/(10 \beta_g \log (g+1))$ flow-well-linked in $H$, by
Lemma~\ref{lem:sparsecut}, $|\delta(S)| \le \beta_g \cdot x(S)/(10
\beta_g \log (g+1)) \le x(S)/(10 \log (g+1))$. We have $a_1 =
x'(S)/2$. Moreover, $x'(S) \ge x(S) - 2 |\delta(S)|$ since a flow path
$p$ crossing the cut with flow $f_p$ can contribute at most $2f_p$ to
the reduction of the marginal values of $x(S)$ at its end points.
Thus $2a_1 = x'(S) \ge x(S) -2 |\delta(S)| \ge (10 \log (g+1) - 2)
|\delta(S)| \ge 8 |\delta(S)|$ since $g \ge 1$, which implies
that $|\delta(S)| \le a_1/4$, as claimed.

Therefore we have the following recurrence for the total
number of edges cut in the overall decomposition:

$$L(a,r) \le L(a_1,r_1) + L(a_2,r_2) + \frac{1}{4 \log (g+1)} \min\{a_1,a_2\}.$$

We prove by induction on $r$ and number of nodes of $G$ that for $r
\le g$, $L(a,r) \le \frac{\log (r+1)}{2\log (g+1)} \cdot a$. We had
already seen the base case with $r = 0$ since $L(a,r) = 0$.  By
induction $L(a_1,r_1) \le \frac{\log (r_1+1)}{2\log (g+1)} \cdot a_1$
and $L(a_2,r_2) \le \frac{\log (r_2+1)}{2\log (g+1)} \cdot a_2$.
Since $r_1+r_2 \le r$, $\min\{r_1+1,r_2+1\} \le (r+1)/2$. It is not
hard to see that $a_1 \log (r_1+1) + a_2 \log (r_2+1) \le (a_1+a_2)
\log (r+1) - \min\{a_1,a_2\} \log 2$. Using the recurrence.
\begin{eqnarray*}
  L(a,r) & \le & L(a_1,r_1) + L(a_2,r_2) + \frac{1}{4 \log (g+1)} \min\{a_1,a_2\} \\
  & \le & \frac{\log (r_1+1)}{2\log (g+1)} \cdot a_1 + \frac{\log (r_2+1)}{2\log (g+1)} \cdot a_2 + \frac{1}{4 \log (g+1)} \min\{a_1,a_2\} \\
  & \le &  \frac{\log (r+1)}{2 \log (g+1)} (a_1 + a_2)  - \frac{\log 2}{2\log (g+1)} \min\{a_1,a_2\} + \frac{1}{4 \log (g+1)} \min\{a_1,a_2\} \\
  & \le & \frac{\log (r+1)}{2 \log (g+1)} \cdot a.
\end{eqnarray*}

Thus $L(\val{\bx},g) \le \frac{\log (g+1)}{2 \log (g+1)} \val{\bx} \le
\val{\bx}/2$.  Hence the total flow that remains after the
decomposition is at least $\val{\bx}/2$.

\subsection{Proof of Theorem~\ref{thm:well-linked-genus}}
\label{subsec:well-linked-genus}

We start with a grouping technique from \cite{CKS-well-linked} that
boosts the well-linkedness.

\begin{theorem}
  \label{thm:boost}  Let $\bx$ be a feasible fractional solution to an
  \medp instance $(G,X,M)$.  Moreover, suppose the terminal set $X$ is
  $\pi$-flow-well-linked in $G$ where $\pi(v) = \rho \cdot x(v)$ for
  some scalar $\rho \le 1$. Then there is a polynomial-time algorithm
  that routes $\Omega(\rho \cdot \val{\bx})$ pairs of $M$ edge-disjointly or
  outputs a new instance $(G,X',M')$ where $M' \subset M$ and $X'$ is
  well-linked and $|M'| = \Omega(\rho |M|)$.
\end{theorem}

Using the preceding theorem we assume that we are working with an
instance $(G,X,M)$ where $X$ is well-linked. Recall that $G$ has genus
at most $g$ and degree of each node is at most $4$.  We use the
observation below (see \cite{Chekuri04b}) that relates the
size of a well-linked set and the treewidth.
\begin{lemma}
  Let $G$ be a graph with maximum degree $\Delta$ and let $X \subseteq V$ be
  a well-linked set in $G$. Then the treewidth of $G$ is $\Omega(|X|/\Delta)$.
\end{lemma}

Thus we can assume that $G$ has treewidth $\Omega(|X|)$.  Demaine
\etal \cite{Demaineetal-bidimensional} showed the following theorem on the size of a grid minor in
graphs of genus $g$ following the work of Robertson and Seymour.

\begin{theorem}[\cite{Demaineetal-bidimensional}]
  Let $G$ be a graph of genus at most $g$. Then $G$ has a grid minor
  of size $\Omega(h/g)$ where $h$ is the treewidth of $G$.
\end{theorem}

Following the scheme from \cite{CKS-well-linked}, one can use the grid
minor as a cross bar to route a large number of pairs from $M$
provided we can route $\Omega(|X|/g)$ terminals to the ``interface''
of the grid-minor of size $\Omega(|X|/g)$ that is guaranteed to exist
in $G$. We can view the grid minor as rows and columns.  In our
current context we take every other node in the first row of the
grid as the interface of the grid-minor. Each node $v$ of the minor
corresponds to a subset of connected nodes $A_v$ in original graph $G$ that
are contracted to form $v$. To simplify
notation we say that $S \subset V$ is the interface of a grid-minor
in $G$ if $|S \cap A_v| = 1$ for each interface node $v$ of the
grid-minor. The following lemma is essentially implicit in previous
work, in particular \cite{CKS-planar} used in the context of planar
graphs.

\begin{lemma}
  \label{lem:grid-routing}
  Suppose $G=(V,E)$ contains a $h \times h$ grid as a minor. Let $S
  \subseteq V$ be the interface of the grid-minor. Then $S$ is
  well-linked in $G$. Moreover, any matching $M$ on $S$ is routable
  in $G$ with congestion $2$.
\end{lemma}

The key technical difficulty is to ensure that $G$ contains
a large grid-minor whose interface is reachable from the terminals $X$.
In a sense, the grid-minor's existence is shown via the
well-linkedness of $X$ and hence there should be such a ``reachable''
grid-minor. The following theorem formalizes the existence
of the desired grid-minor.

\begin{theorem}
  \label{thm:reachable-grid}
  Let $G$ be a graph of genus $g > 0$. Suppose $X$ is a well-linked
  set in $G$. Then $G$ contains a $h \times h$ grid-minor
  with interface $S$ such that $h = \Omega(|X|/g)$ and
  at least $h/8$ edge-disjoint paths in $G$ from $X$ to $S$.
\end{theorem}

Proving the preceding theorem formally 
requires work. In \cite{CKS-planar} an argument tailored to planar
graphs was used to prove a similar theorem, and in fact a
polynomial-time algorithm was given to find the desired grid-minor. In that
same paper a general deletable edge lemma \cite{cks04} was announced (though
never published); in the appendix we give  a streamlined
proof based on the original manuscript. We
postpone the proof of the preceding theorem and outline how to
complete the proof of Theorem~\ref{thm:well-linked-genus}.
We start with our well-linked set $X$, and via Theorem~\ref{thm:reachable-grid},
find a grid-minor of size $h = \Omega(|X|/g)$ such that there
are $h/8$ edge-disjoint paths from $X$ to the interface $S$.
Let $X' \subset X$ be $h/8$ terminals that are the end points of these
edge-disjoint paths. Recall that we are interested in routing
a given matching $M$ on $X$. If $X'$ contains a ``large'' sub-matching
$M' \subset M$ then we can use the grid-minor to route $M'$ with
congestion $3$ as follows. Let $S' \subset S$ be end points of
the paths from $X'$ to $S$. Clearly $M'$ induces a matching on
$S'$ and the grid-minor can route this matching with congestion $2$.
Patching this routing with the paths from $X'$ to $S'$ gives the
desired congestion $3$ routing of $M'$ in $G$. However, it may be the
case that $|M'|$ is very small, even zero,
because $X'$ only contains one end point from every edge in $M$.
However, here, we can use the fact that $X$ and $S$
are  well-linked in $G$ to argue that we can route any desired
subset of $X$ of sufficiently large size to $S$. Thus, we can assume
that indeed $|M'|$ is a constant fraction of $|X'|$.
This was essentially done in \cite{Chekuri04b} and details
are also included in the appendix (see Lemma~\ref{lem:alltoall2}).
This shows the number of pairs from $M$ that can be routed
with congestion $3$ in $G$ is a constant fraction of $h$, the size
of the grid-minor. Since $h = \Omega(|X|/g)$ we route
$\Omega(|X|/g)$ pairs from $M$. This finishes the proof of
Theorem~\ref{thm:well-linked-genus}.


Now we come to the proof of Theorem~\ref{thm:reachable-grid}.  This is
done in Section~\ref{sec:deletable-edge} via the following high-level approach.  Suppose
$X$ is a well-linked set in $G$. We are guaranteed that $G$ has a
grid-minor of size $h = \Omega(|X|/g)$. Let $S$ be its interface. If
there are $h/8$ edge-disjoint paths from $X$ to $S$ then we are
done. Otherwise the claim is that there is an edge $e$ such that $X$
is well-linked in $G-e$. This is established in
Theorem~\ref{thm:deletable}.  In other words, if we start with a graph
which is edge-minimal subject to $X$ being well-linked, then following
the procedure above  yields a grid-minor that is reachable from
$X$. The reason that this does not immediately lead to a polynomial-time
algorithm to find such a grid-minor is the fact that checking whether
a given set $X$ is not well-linked is NP-Complete. Note that we can
check if $X$ if flow-well-linked but the deletable edge lemma we have
works with the notion of cut-well-linkedness which is hard to check.
In \cite{CKS-planar} a poly-time deletable edge lemma was obtained
for the special case of planar graphs and we believe that it can be
extended to the case of genus $g$ graphs but the technical details
are quite involved.

\section{Open Problems and Concluding Remarks}
Resolving Conjecture~\ref{conj:minor-free} is the main open problem
that arises from this work. In particular, does a planar graph with a
single vortex satisfy the CFCC property? This is
is the key technical obstacle.

Theorem~\ref{thm:ksum-main} loses an approximation factor that is
exponential in $k$. Is this necessary? In particular, the known
integrality gap for the flow LP on graphs of treewidth $k$ is only
$\Omega(k)$; this comes from the grid example. Moreover, the
congestion we obtain for $\cG_k$ is $\beta+3$ where $\beta$ is the
congestion guaranteed for the class $\cG$. Can this be improved
further, say to $\beta$ or $\beta+1$?

We believe that Theorem~\ref{thm:genus_main} can be made algorithmic
and also conjecture that the congestion bound can be improved to
$2$. The key bottleneck is to find an algorithmic proof of
Theorem~\ref{thm:reachable-grid}.
\fi

\medskip
\noindent
{\bf Acknowledgments:} We thank Sanjeev Khanna for allowing us to
include Section~\ref{sec:deletable-edge} in the paper for the sake
of completeness.

\bibliography{conf}{}
\bibliographystyle{plain}

\appendix
\ifabstract
\section{Figures for Proof of Theorem~\ref{thm:ksum-main}}

\begin{figure}[thb]
\begin{tabular}{cc}

\begin{minipage}{0.48\textwidth}

\begin{tikzpicture}[x=0.17cm,y=0.17cm,>=latex]

\vertex{black} (nx32y15) at (32,15) {};
\vertex{black} (nx34y19) at (34,19) {};
\vertex{black} (nx36y15) at (36,15) {};
\vertex{black} (nx38y18) at (38,18) {};
\vertex{black} (nx42y21) at (42,21) {};
\vertex{black} (nx44y24) at (44,24) {};
\vertex{black} (nx43y27) at (43,27) {};
\vertex{black} (nx40y25) at (40,25) {};
\vertex{black} (nx41y35) at (41,35) {};
\vertex{black} (nx40y38) at (40,38) {};
\vertex{black} (nx35y38) at (35,38) {};
\vertex{black} (nx37y35) at (37,35) {};
\vertex{black} (nx13y27) at (13,27) {};
\vertex{black} (nx16y24) at (16,24) {};
\vertex{black} (nx15y31) at (15,31) {};
\vertex{black} (nx25y39) at (25,39) {};
\vertex{black} (nx25y35) at (25,35) {};
\vertex{black} (nx21y35) at (21,35) {};
\vertex{black} (nx20y39) at (20,39) {};

\Vertex{black} (nx34y27) at (34,27) {};
\Vertex{black} (nx28y26) at (28,26) {};
\Vertex{black} (nx29y21) at (29,21) {};
\Vertex{black} (nx36y23) at (36,23) {};
\Vertex{black} (nx38y30) at (38,30) {};
\Vertex{black} (nx33y33) at (33,33) {};
\Vertex{black} (nx27y31) at (27,31) {};
\Vertex{black} (nx21y30) at (21,30) {};
\Vertex{black} (nx22y25) at (22,25) {};

\fill[black,nearly transparent]
      (29.19,20.33) arc[start angle = 286, end angle = 140, radius = 0.70]
   -- (33.46,27.45) arc[start angle = 140, end angle =  27, radius = 0.70]
   -- (36.63,23.31) arc[start angle =  27, end angle = -74, radius = 0.70]
   -- cycle;
\draw[black] (nx34y27) -- (nx29y21);
\draw[black] (nx34y27) -- (nx32y15);
\draw[black] (nx32y15) -- (nx38y18);
\draw[black] (nx36y15) -- (nx34y19);
\draw[black] (nx36y15) -- (nx38y18);
\draw[black] (nx32y15) -- (nx36y15);
\draw[black] (nx32y15) -- (nx29y21);
\draw[black] (nx34y19) -- (nx32y15);
\draw[black] (nx38y18) -- (nx34y19);
\draw[black] (nx36y23) -- (nx38y18);
\draw[black] (nx34y19) -- (nx36y23);
\draw[black] (nx29y21) -- (nx34y19);

\fill[black,nearly transparent] 
      (37.46,30.72) arc[start angle = 127, end angle = -16, radius = 0.90]
   -- (36.87,22.75) arc[start angle = 344, end angle = 207, radius = 0.90]
   -- (33.20,26.60) arc[start angle = 207, end angle = 127, radius = 0.90]
   -- cycle;
\draw[black] (nx40y25) -- (nx34y27);
\draw[black] (nx43y27) -- (nx34y27);
\draw[black] (nx40y25) -- (nx44y24);
\draw[black] (nx40y25) -- (nx42y21);
\draw[black] (nx43y27) -- (nx40y25);
\draw[black] (nx44y24) -- (nx43y27);
\draw[black] (nx42y21) -- (nx44y24);
\draw[black] (nx36y23) -- (nx42y21);
\draw[black] (nx40y25) -- (nx36y23);
\draw[black] (nx38y30) -- (nx40y25);
\draw[black] (nx38y30) -- (nx43y27);

\fill[black,nearly transparent] (35.5, 31.5) circle[x radius = 3.5, y radius = 1.3, rotate = -31] {};
\draw[black] (nx41y35) -- (nx40y38);
\draw[black] (nx37y35) -- (nx41y35);
\draw[black] (nx35y38) -- (nx37y35);
\draw[black] (nx40y38) -- (nx35y38);
\draw[black] (nx37y35) -- (nx40y38);
\draw[black] (nx37y35) -- (nx38y30);
\draw[black] (nx33y33) -- (nx37y35);
\draw[black] (nx35y38) -- (nx41y35);
\draw[black] (nx38y30) -- (nx35y38);
\draw[black] (nx41y35) -- (nx33y33);
\draw[black] (nx38y30) -- (nx41y35);
\draw[black] (nx33y33) -- (nx35y38);

\fill[black,nearly transparent]
      (22.54,24.55) arc[start angle = 320, end angle = 191, radius = 0.70]
   -- (20.31,29.86) arc[start angle = 191, end angle =  99, radius = 0.70]
   -- (26.88,31.69) arc[start angle =  99, end angle = -40, radius = 0.70]
   -- cycle;
\draw[black] (nx25y35) -- (nx22y25);
\draw[black] (nx21y35) -- (nx20y39);
\draw[black] (nx25y35) -- (nx25y39);
\draw[black] (nx21y35) -- (nx25y35);
\draw[black] (nx25y39) -- (nx21y35);
\draw[black] (nx20y39) -- (nx25y39);
\draw[black] (nx25y35) -- (nx20y39);
\draw[black] (nx27y31) -- (nx25y35);
\draw[black] (nx21y35) -- (nx27y31);
\draw[black] (nx21y30) -- (nx21y35);

\fill[black,nearly transparent] (21.5, 27.5) circle[x radius = 3.05, y radius = 1.3, rotate = -78.7] {};
\draw[black] (nx13y27) -- (nx15y31);
\draw[black] (nx16y24) -- (nx13y27);
\draw[black] (nx15y31) -- (nx16y24);
\draw[black] (nx21y30) -- (nx15y31);
\draw[black] (nx22y25) -- (nx16y24);
\draw[black] (nx13y27) -- (nx22y25);
\draw[black] (nx21y30) -- (nx13y27);

\draw[very thick, black!40]
      (21.21,23.61) arc[start angle = 240, end angle = 191, radius = 1.60]
   -- (19.43,29.69) arc[start angle = 191, end angle = 104, radius = 1.60]
   -- (32.61,34.55) arc[start angle = 104, end angle =  59, radius = 1.60]
   -- (38.82,31.37) arc[start angle =  59, end angle = -16, radius = 1.60]
   -- (37.54,22.56) arc[start angle = -16, end angle = -74, radius = 1.60]
   -- (29.44,19.46) arc[start angle = 286, end angle = 240, radius = 1.60]
   -- cycle
;

\draw[black,thick] (nx36y23) -- (nx34y27);
\draw[black,thick] (nx29y21) -- (nx36y23);
\draw[black,thick] (nx29y21) -- (nx22y25);
\draw[black,thick] (nx28y26) -- (nx29y21);
\draw[black,thick] (nx22y25) -- (nx21y30);
\draw[black,thick] (nx28y26) -- (nx22y25);
\draw[black,thick] (nx28y26) -- (nx34y27);
\draw[black,thick] (nx27y31) -- (nx28y26);
\draw[black,thick] (nx34y27) -- (nx33y33);
\draw[black,thick] (nx38y30) -- (nx34y27);
\draw[black,thick] (nx27y31) -- (nx33y33);
\draw[black,thick] (nx21y30) -- (nx27y31);

\path (12,29) node {$G_1$};
\path (22.5,40) node {$G_2$};
\path (41,39) node {$G_3$};
\path (44,21) node {$G_4$};
\path (34,13.5) node {$G_5$};

\end{tikzpicture}

\end{minipage} &

\begin{minipage}{0.48\textwidth}

\begin{tikzpicture}[x=0.17cm,y=0.17cm,>=latex]

\vertex{black!10} (nx32y15) at (32,15) {};
\vertex{black!10} (nx34y19) at (34,19) {};
\vertex{black!10} (nx36y15) at (36,15) {};
\vertex{black!10} (nx38y18) at (38,18) {};
\vertex{black!10} (nx42y21) at (42,21) {};
\vertex{black!10} (nx44y24) at (44,24) {};
\vertex{black!10} (nx43y27) at (43,27) {};
\vertex{black!10} (nx40y25) at (40,25) {};
\vertex{black!10} (nx41y35) at (41,35) {};
\vertex{black!10} (nx40y38) at (40,38) {};
\vertex{black!10} (nx35y38) at (35,38) {};
\vertex{black!10} (nx37y35) at (37,35) {};
\vertex{black!10} (nx13y27) at (13,27) {};
\vertex{black!10} (nx16y24) at (16,24) {};
\vertex{black!10} (nx15y31) at (15,31) {};
\vertex{black!10} (nx25y39) at (25,39) {};
\vertex{black!10} (nx25y35) at (25,35) {};
\vertex{black!10} (nx21y35) at (21,35) {};
\vertex{black!10} (nx20y39) at (20,39) {};

\Vertex{black} (nx34y27) at (34,27) {};
\Vertex{black} (nx28y26) at (28,26) {};
\Vertex{black} (nx29y21) at (29,21) {};
\Vertex{black} (nx36y23) at (36,23) {};
\Vertex{black} (nx38y30) at (38,30) {};
\Vertex{black} (nx33y33) at (33,33) {};
\Vertex{black} (nx27y31) at (27,31) {};
\Vertex{black} (nx21y30) at (21,30) {};
\Vertex{black} (nx22y25) at (22,25) {};

\fill[black,very nearly transparent]
      (29.19,20.33) arc[start angle = 286, end angle = 140, radius = 0.70]
   -- (33.46,27.45) arc[start angle = 140, end angle =  27, radius = 0.70]
   -- (36.63,23.31) arc[start angle =  27, end angle = -74, radius = 0.70]
   -- cycle;
\draw[black,very nearly transparent] (nx34y27) -- (nx29y21);
\draw[black,very nearly transparent] (nx34y27) -- (nx32y15);
\draw[black,very nearly transparent] (nx32y15) -- (nx38y18);
\draw[black,very nearly transparent] (nx36y15) -- (nx34y19);
\draw[black,very nearly transparent] (nx36y15) -- (nx38y18);
\draw[black,very nearly transparent] (nx32y15) -- (nx36y15);
\draw[black,very nearly transparent] (nx32y15) -- (nx29y21);
\draw[black,very nearly transparent] (nx34y19) -- (nx32y15);
\draw[black,very nearly transparent] (nx38y18) -- (nx34y19);
\draw[black,very nearly transparent] (nx36y23) -- (nx38y18);
\draw[black,very nearly transparent] (nx34y19) -- (nx36y23);
\draw[black,very nearly transparent] (nx29y21) -- (nx34y19);

\fill[black,very nearly transparent] 
      (37.46,30.72) arc[start angle = 127, end angle = -16, radius = 0.90]
   -- (36.87,22.75) arc[start angle = 344, end angle = 207, radius = 0.90]
   -- (33.20,26.60) arc[start angle = 207, end angle = 127, radius = 0.90]
   -- cycle;
\draw[black,very nearly transparent] (nx40y25) -- (nx34y27);
\draw[black,very nearly transparent] (nx43y27) -- (nx34y27);
\draw[black,very nearly transparent] (nx40y25) -- (nx44y24);
\draw[black,very nearly transparent] (nx40y25) -- (nx42y21);
\draw[black,very nearly transparent] (nx43y27) -- (nx40y25);
\draw[black,very nearly transparent] (nx44y24) -- (nx43y27);
\draw[black,very nearly transparent] (nx42y21) -- (nx44y24);
\draw[black,very nearly transparent] (nx36y23) -- (nx42y21);
\draw[black,very nearly transparent] (nx40y25) -- (nx36y23);
\draw[black,very nearly transparent] (nx38y30) -- (nx40y25);
\draw[black,very nearly transparent] (nx38y30) -- (nx43y27);

\fill[black,very nearly transparent] (35.5, 31.5) circle[x radius = 3.5, y radius = 1.3, rotate = -31] {};
\draw[black,very nearly transparent] (nx41y35) -- (nx40y38);
\draw[black,very nearly transparent] (nx37y35) -- (nx41y35);
\draw[black,very nearly transparent] (nx35y38) -- (nx37y35);
\draw[black,very nearly transparent] (nx40y38) -- (nx35y38);
\draw[black,very nearly transparent] (nx37y35) -- (nx40y38);
\draw[black,very nearly transparent] (nx37y35) -- (nx38y30);
\draw[black,very nearly transparent] (nx33y33) -- (nx37y35);
\draw[black,very nearly transparent] (nx35y38) -- (nx41y35);
\draw[black,very nearly transparent] (nx38y30) -- (nx35y38);
\draw[black,very nearly transparent] (nx41y35) -- (nx33y33);
\draw[black,very nearly transparent] (nx38y30) -- (nx41y35);
\draw[black,very nearly transparent] (nx33y33) -- (nx35y38);

\fill[black,very nearly transparent]
      (22.54,24.55) arc[start angle = 320, end angle = 191, radius = 0.70]
   -- (20.31,29.86) arc[start angle = 191, end angle =  99, radius = 0.70]
   -- (26.88,31.69) arc[start angle =  99, end angle = -40, radius = 0.70]
   -- cycle;
\draw[black,very nearly transparent] (nx25y35) -- (nx22y25);
\draw[black,very nearly transparent] (nx21y35) -- (nx20y39);
\draw[black,very nearly transparent] (nx25y35) -- (nx25y39);
\draw[black,very nearly transparent] (nx21y35) -- (nx25y35);
\draw[black,very nearly transparent] (nx25y39) -- (nx21y35);
\draw[black,very nearly transparent] (nx20y39) -- (nx25y39);
\draw[black,very nearly transparent] (nx25y35) -- (nx20y39);
\draw[black,very nearly transparent] (nx27y31) -- (nx25y35);
\draw[black,very nearly transparent] (nx21y35) -- (nx27y31);
\draw[black,very nearly transparent] (nx21y30) -- (nx21y35);

\fill[black,very nearly transparent] (21.5, 27.5) circle[x radius = 3.05, y radius = 1.3, rotate = -78.7] {};
\draw[black,very nearly transparent] (nx13y27) -- (nx15y31);
\draw[black,very nearly transparent] (nx16y24) -- (nx13y27);
\draw[black,very nearly transparent] (nx15y31) -- (nx16y24);
\draw[black,very nearly transparent] (nx21y30) -- (nx15y31);
\draw[black,very nearly transparent] (nx22y25) -- (nx16y24);
\draw[black,very nearly transparent] (nx13y27) -- (nx22y25);
\draw[black,very nearly transparent] (nx21y30) -- (nx13y27);

\draw[very thick, black!40]
      (21.21,23.61) arc[start angle = 240, end angle = 191, radius = 1.60]
   -- (19.43,29.69) arc[start angle = 191, end angle = 104, radius = 1.60]
   -- (32.61,34.55) arc[start angle = 104, end angle =  59, radius = 1.60]
   -- (38.82,31.37) arc[start angle =  59, end angle = -16, radius = 1.60]
   -- (37.54,22.56) arc[start angle = -16, end angle = -74, radius = 1.60]
   -- (29.44,19.46) arc[start angle = 286, end angle = 240, radius = 1.60]
   -- cycle
;

\draw[black, dashed,very thin] (nx29y21) -- (nx22y25);
\draw[black, dashed,very thin] (nx28y26) -- (nx29y21);
\draw[black, dashed,very thin] (nx28y26) -- (nx22y25);
\draw[black, dashed,very thin] (nx28y26) -- (nx34y27);
\draw[black, dashed,very thin] (nx27y31) -- (nx28y26);
\draw[black, dashed,very thin] (nx34y27) -- (nx33y33);
\draw[black, dashed,very thin] (nx27y31) -- (nx33y33);

\draw[black,thick] (nx38y30) -- (nx34y27);
\draw[black,thick] (nx36y23) -- (nx34y27);
\draw[black,thick] (nx29y21) -- (nx34y27);
\draw[black,thick] (nx21y30) -- (nx27y31);
\draw[black,thick] (nx22y25) -- (nx21y30);
\draw[black,thick] (nx29y21) -- (nx36y23);
\draw[black,thick] (nx33y33) -- (nx38y30);
\draw[black,thick] (nx38y30) -- (nx36y23);
\draw[black,thick] (nx22y25) -- (nx27y31);

\end{tikzpicture}

\end{minipage}
\end{tabular}
\caption{Illustration of the base case with $k=0$ and $p = 3$: the graph
corresponding to each degenerate leaf is replaced by a sparsifier on the
associated separator.}
\label{fig:base-case}
\label{fig:base-case}
\end{figure}

\begin{figure}[hbt]
\label{fig:induction-step}
\begin{tabular}{cc}

\begin{minipage}{0.48\textwidth}

\begin{tikzpicture}[x=0.4cm,y=0.4cm,>=latex]

\nellipse{black,nearly transparent} (a) at (8,10) {};
\nellipse{black,nearly transparent} (b) at (4,7) {};
\nellipse{black,nearly transparent} (c) at (12,7) {};
\nellipse{black,nearly transparent} (d) at (2,4) {};
\nellipse{black,nearly transparent} (g) at (14,4) {};

\nellipse{black,nearly transparent} (e) at (6,4) {};
\nellipse{black,nearly transparent} (f) at (10,4) {};
\nellipse{black,nearly transparent} (h) at (4,1) {};
\nellipse{black,nearly transparent} (i) at (8,1) {};

\draw[dashed] (11,12) -- (a.north);
\draw
  (a.south west) -- (b.north)
  (a.south east) -- (c.north)
  (b.south west) -- (d.north)
  (c.south east) -- (g.north)
;

\draw
  (b.south east) -- (e.north)
  (c.south west) -- (f.north)
  (e.south west) -- (h.north)
  (e.south east) -- (i.north)
;

\draw[very thick,black] 
  (6,-1) .. controls (13,-1) and (12,4) ..
  (12,6) .. controls (12,8) and (10,10) ..
  (8,10) .. controls (6,10) and (4,8) ..
  (3.5,7) .. controls (3,6) and (0,1) ..
  (0.5,0) .. controls (1,-1) and (1,-1) .. (6,-1);

\path[rounded corners=24pt,fill=black,very nearly transparent]
  (6,6) -- (1.5,0) -- (10.5,0) -- cycle;
\path[rounded corners=24pt,draw=black,dotted]
  (6,6) -- (1.5,0) -- (10.5,0) -- cycle;
\node (t1) at (6,2.5) {$T_1$};

\path[rounded corners=24pt,fill=black,very nearly transparent]
 (10,6.4) -- (7.6,2.8) -- (12.4,2.8) -- cycle;
\path[rounded corners=24pt,draw=black,dotted]
 (10,6.4) -- (7.6,2.8) -- (12.4,2.8) -- cycle;
\node (t2) at (10,4) {$T_2$};

\node at (8,8) {{\large $U$}};

\end{tikzpicture}

\end{minipage} &

\begin{minipage}{0.48\textwidth}

\begin{tikzpicture}[x=0.5cm,y=0.5cm,>=latex]

\path[clip] 
  (6,-1) .. controls (13,-1) and (12,4) ..
  (12,6) .. controls (12,8) and (10,10) ..
  (8,10) .. controls (6,10) and (4,8) ..
  (3.5,7) .. controls (3,6) and (0,1) ..
  (0.5,0) .. controls (1,-1) and (1,-1) .. (6,-1);

\nellipse{black,nearly transparent} (a) at (8,10) {};
\nellipse{black,nearly transparent} (b) at (4,7) {};
\nellipse{black,nearly transparent} (c) at (12,7) {};
\nellipse{black,nearly transparent} (d) at (2,4) {};


\draw[dashed] (11,12) -- (a.north);
\draw
  (a.south west) -- (b.north east)
  (a.south east) -- (c.north west)
  (b.south) -- (d.north east)
  (10,6.4) ++(0,-12pt) -- (c.south west)
;

\draw
  (6,6) ++(0,-11pt)-- (b.south east);

\path[rounded corners=24pt,fill=black,nearly transparent]
  (6,6) -- (1.5,0) -- (10.5,0) -- cycle;
\node (t1) at (6,2) {degenerate leaf};

\path[rounded corners=24pt,fill=black,nearly transparent]
 (10,6.4) -- (7.6,2.8) -- (12.4,2.8) -- cycle;

\node[align=center] (t2) at (10,3.8) {degenerate\\ leaf};

\node at (8,8.3) {{\large $\mathcal{T}'_U$}};

\end{tikzpicture}

\end{minipage}
\end{tabular}
\caption{Induction step: reducing the effective treewidth of $G[U]$
by creating degenerate leaves. $T_1,T_2$ are maximal subtrees rooted at separators of size $k$, converted to degenerate leaves.}
\label{fig:induction-step}
\end{figure}

\newpage
\section{Technical Ingredients}
\label{sec:tools}
We rely on several technical tools and ingredients that are either explicitly
or implicitly used in recent work on \medp.

\subsection{Moving Terminals}
We first describe a general tool (ideas of which are leveraged also in
previous work, cf. \cite{CKS-planar-constant,CKS-treewidth}) that allows
us to reduce a \medp instance to a simpler one by
moving the terminals to a specific set of new locations (nodes). The
two instances are equivalent for \medp, up to an additional constant
congestion and constant factor approximation.


\begin{lemma}
\label{lem:rr}
Suppose we have a (matching) instance $(G,H)$ of \medp{} with some
solution $\bx$ to its LP relaxation. Let $val(\bx) = \sum_P x_P$. Suppose
that for some $S,R \subseteq V(G)$, there is a flow which routes
$x(s)$ from each $s \in S$, and all flow terminates in $R$.  Then
there is another (matching) instance $(G',H')$ with the following
properties.
\begin{enumerate}\setlength{\itemsep}{0pt}
\item The new instance has a (fractional) solution of value at least
  $val(\bx)/5$,
\item If there is an integral solution for $(G',H')$ of
  congestion $c$, then there is an integral solution for $(G,H)$
  of the same value and congestion $c+2$,
\item $G',H'$ is obtained from $G,H$ by hanging off pendant stars from
  some of the nodes.
\end{enumerate}
\end{lemma}

\begin{proof}
Let $T$ be a forest of $G$ spanning all the nodes of $S$ such that each
component of $T$ contains at least one node of $R$. We consider each
component of $T$ separately, so we assume here that $T$ is a tree. Let
$r \in R \cap V(T)$ and take this as a root of $T$.

We partition $S$ into subsets $S_1,\ldots,S_\ell$ with the following properties.
\begin{itemize}\setlength{\itemsep}{0pt}
\item[$(i)$] for each $i \in [1,\ell]$, $1 \leq x(S_i) \leq 2$ (except
  possibly $S_1$ may have $x(S_1) < 1$),
\item[$(ii)$] to each $S_i$ is associated a subtree $T_i$ of $T$ spanning $S_i$,
\item[$(iii)$] $T_1,\ldots, T_l$ are edge-disjoint.
\end{itemize}
We achieve this by  the following iterative scheme.
\begin{itemize}\setlength{\itemsep}{0pt}
\item If $x(T) \leq 2$, choose $S_1 = S$, $T_1 = T$, $\ell = 1$. Else:
\item Find a deepest node $v$ in $T$, such that the subtree $T'$
  rooted at $v$ has $x(T') \geq 1$.
\item let $v_1,\ldots,v_i$ be the children of $v$, and $T'_1,\ldots,T'_i$
  be the subtrees rooted at $v_1,\ldots,v_i$ respectively. If
  $\sum_{j=1}^i x(T'_j) < 1$, then define $A \eqdef E(T')$ and $U
  \eqdef V(T') \cap S$.  Otherwise, find the smallest $i'$ such that
  $x(B) \geq 1$, where $B \eqdef \bigcup_{j=1}^{i'} V(T'_j) \cap S$.
  Then set $A \eqdef \bigcup_{j=1}^{i'} (E(T'_j) \cup \{vv_j\})$. In
  both cases, we get $1 \leq x(A) \leq 2$, and $A$ induces a tree.
  (Note that in the latter case $v$ was not placed in $B$ but does lie
  in $V(A)$.)
\item proceed inductively on $T - A$ and $S - B$, to find $S_1,\ldots
  S_{\ell'}$ and $T_1,\ldots T_{\ell'}$. Then $\ell \eqdef \ell' + 1$, $S_\ell \eqdef B$
  and $T_\ell \eqdef A$.
\end{itemize}
Note that with this scheme, $x(S_1)$ might be less than $1$, but in
that case $T_1$ contains the root node from $R$. As the hypotheses
$x(S_i) \geq 1$ is only used to prove the existence of edge-disjoint
paths from the $S_i$'s to $R$, this does not pose a problem (we can
simply take the trivial path at $r$ for $S_1$).

Each $S_i$ is called a \emph{cluster}. As each cluster sends at least
one unit of flow to $R$ (with the exception of $S_1$ already
mentioned), there is a family of edge-disjoint paths $P_1,\ldots,P_l$,
where each $P_i$ goes from a node $s_i \in S_i$ to some node $r_i \in
R$.  This can be seen by adding dummy source nodes (one for each
$S_i$) adjacent to nodes in each $S_i$, and a single dummy sink node
adjacent from each $r \in R$ (a detailed proof is found in \cite{CKS-planar-constant}).

 We now define a new instance of \medp{} $G',H'$. $G'$ is obtained
 from $G$ by adding $\ell$ new nodes $u_1,\ldots,u_\ell$ with degree one,
 where $u_i$ is adjacent to $r_i$. The capacity of a new edge $r_iu_i$
 is $1$, and we re-define $P_i$ as extending to $u_i$.  We identify
 each terminal in $S$ with the $u_i$ associated with its cluster as
 follows. Let $\phi(s) \eqdef s$ if $s \notin S$ and $\phi(s) = u_i$
 if $s \in S_i$. Then let $\phi(H) \eqdef\{\phi(s)\phi(t)~:~st \in
 H\}$. These demands do not yet form a matching, so $H'$ is obtained
 from $\phi(H)$ by simply deporting each of the terminals in $S$ to
 new nodes forming leaves.

  We show how to transform $\bx$ into a fractional flow $\bx'$ in $G',H'$
  with congestion $5$, such that $\bx'$ has the same value as $\bx$. For
  that, we only extend the flow paths for the demands in $H' \setminus
  H$. Let $st \in H$ be such a demand and $s't' = \phi(s)\phi(t)$ its
  image. For any $st$-path $P$ with value $x_P$, let $P'$ be the path
  obtained from $P$ by:
\begin{itemize}\setlength{\itemsep}{0pt}
\item if $s \in S_i$ for some $i$, concatenate $P_i$ and the unique $ss_i$-path of $T_i$,
\item similarly if $t \in S_j$ for some $j$, concatenate $P_j$ and the unique $ss_j$-path of $T_j$.
\end{itemize}

Then, set $x'_{P'} \eqdef x_P$. Then $\bx'$ has the same value as $\bx$ by
construction, but has higher congestion.  Any of the edges $r_iu_i$
(or additional leaves from an $u_i$) have congestion at most $2$ by
construction; thus it is enough to focus on edges within $G$. The
original flow paths incur congestion of at most 1 on any edge, so we
address the added congestion from extending the flow paths. The edges
of any $P_i$ are charged by at most $2$ units (by terminals within
$S_i$) and each $T_i$ is also charge by at most $2$ units. As an edge
may be contained in at most one $P_i$ and at most one $T_i$ the extra
congestion is bounded by $4$. Hence the total congestion of $x'$ is at
most $5$, and in particular, this implies that the fractional optimal
solution in $G'$ is at least $\frac{1}{5} \opt$.

 Suppose now that we have an integral solution to \medp{} for $G',H'$
 with congestion $c$. We show how to transform it into a solution for
 $G,H$ of the same value with congestion $c+2$. Since we can assume
 the flow paths used are simple, we only need to address the flow
 paths for demands in $H' \setminus H$. Let $P$ be any path in the
 solution satisfying a commodity associated with a node $s' =
 \phi(s)$, where $s$ is in $S_i$. Then we extend $P$ by concatenating
 $P_i$ and the unique $s_is$-path of $T_i$ to it.  We may also
 shortcut this to obtain a simple path. Again, this clearly defines a
 solution of same value to the original problem. Since the capacity of
 $r_ix_i$ is one, we use each $P_i$ at most once, and a path in $T_i$
 is used for only one such $s \in S_i$.  As the paths $P_i$ are
 disjoint, and the subtrees are disjoint, each edge is used at
 most $c+2$ times: $c$ for the original routing, $1$ for the $P_i$
 paths, and $1$ for the paths inside the $T_i$'s.
\end{proof}

\subsection{Sparsifiers}

Let $G=(V,E)$ be a (multi) graph and let $S \subset V$. We are
interested in creating a graph $H$ only on the node set $S$ that acts
as a proxy for routing between nodes in $S$ in the original graph $G$.
The notion of sparsifiers, introduced in
\cite{moitra2009approximation}, has many possible formulations depending on the  various applications.
For instance, a Gomory-Hu Tree can be viewed as a sparsifier which encodes pairwise maximum flows in a graph.
We are interested in the following model.
We
say that $H=(S,E_H)$ is a {\em $(\sigma,\rho)$-sparsifier} for $S$ in $G$ if
the following properties are true:
\begin{itemize}
\item any feasible (fractional) multicommodity flow in $G$ with the
  endpoints in $S$ is (fractionally) routable in $H$ with congestion
  at most $\sigma$
\item any integer multicommodity flow  in $H$ is integrally routable
  in $G$ with congestion $\rho$.
\end{itemize}

Existing sparsifier results mostly focus on fractional routing (or cut
preservation) while we need {\em integer} sparsifiers in the sense of
the second point above. Chuzhoy~\cite{Chuzhoy12} developed an integer
sparsifier result but it uses Steiner nodes and has limitations that
preclude its direct use in our setting. Instead, a simple
argument based on splitting-off 
gives the following weak
sparsifier result that suffices for our purposes.

\begin{theorem}
  \label{thm:sparsifier}
  Let $G=(V,E)$ be a graph and $S \subset V$. There
  is a $(|S|^2,2)$-sparsifier for $S$ in $G$.
\end{theorem}

\begin{proof}
  First, by standard $T$-join theory, $G$ contains a subset $E'$ of
  edges, such that if we add an extra copy of each such edge, we
  obtain an Eulerian graph $G'$.  We may now apply splitting off
  repeatedly at the (even degree) nodes of $V-S$.  Each operation
  preserves the minimum cut between any pair of nodes $u,v \in S$.
  This ultimately results in a (multi) graph $H=(S,F)$ on $S$.
  We claim that $H$ is the desired sparsifier.

  Note that any integral routing in $H$ can easily mapped to an
  integral routing in $G'$ since the edges in $F$ map to
  edge-disjoint paths in $G'$.  Since $G'$ had potentially an extra
  copy of an edge from $G$ we see that $\rho = 2$.

  Now consider any fractional multicommodity flow in $G$ between nodes
  in $S$. Say $d(uv)$ flow is routed between $u,v \in S$.  Then
  $d(u,v) \le \lambda_G(u,v)$ where $\lambda_G(u,v)$ is the capacity
  of a min $u$-$v$ cut in $G$. Since the splitting-off operation
  preserved connectivity, $\lambda_H(u,v) \geq \lambda_G(u,v)$, hence
  we can route $d(u,v)$ flow between $u$ and $v$ in $H$.
  However, we have $|S|(|S|-1)/2$ distinct pairs of nodes in $S$
  and routing their flows simultaneously in $H$ can result
  in a congestion of at most $|S|(|S|-1) \le |S|^2$ since each
  individual flow can be feasibly routed in $H$. This shows that
  $\sigma \le |S|^2$.
\end{proof}

\begin{remark}
  \label{rem:sparsifier}
  The proof of the preceding theorem shows that the congestion
  parameter $\rho$ can be chosen to be an additive $1$ if $G$ is a
  capacitated graph.
\end{remark}

\subsection{Routings through a small set of nodes}

We use as a black box the following result from Section 3.1\ in
\cite{CKS-sqrtn}.
\begin{prop}
\label{prop:oldsinglenode}
Let $G,H$ be a \medp instance and let $\bx$ be a fractional solution
such that there is a node $v$ that is contained in every flow path
with positive flow. Then there is a polynomial time algorithm
that routes at least $\frac{1}{12}\sum_i x_i$ pairs
from $H$ on edge-disjoint paths.
\end{prop}

\begin{remark}
  The bound of $1/12$ in the preceding proposition is not explicitly
  stated in \cite{CKS-sqrtn} but can be inferred from the arguments.
\end{remark}


Now suppose that instead of a single node $v$, there is a subset $S$
that intersects every flow path in a fractional solution $\bx$. It is
then easy to see that there is a node $v$ that intersects flow paths
of total value at least $\sum_i x_i/|S|$. We can then apply the preceding
proposition to claim that we can route $\frac{1}{12|S|}\sum_i x_i$ pairs.
We combine this with a simple re-routing argument that is relevant
to our algorithm to obtain the following.

\begin{prop}
\label{prop:withreroute}
Let $G,H$ be a matching instance of \medp and let $\bx$ be a feasible
fractional solution for it.  Suppose that there is also a second flow
that routes at least $x_i/\alpha$ flow from each terminal to some $S
\subseteq V$ where $\alpha \ge 1$. Then there is an integral routing
of at least $\frac{1}{36\alpha |S|}\sum_i x_i$ pairs.
\end{prop}

\begin{proof}
  Let $v \in S$ be the terminal which receives the most flow. Clearly
  this is of value at least $\frac{\sum_i x_i}{\alpha |S|}$. Consider
  a pair $s_it_i$ such that one of the end points, say $s_i$ sends
  $y_i \le x_i/\alpha$ flow to $v$. The other end point $t_i$ may send
  less than $y_i$ (or no flow) to $v$. We may then create a $y_i$ flow
  from $t_i$ to $v$ by using the $x_i$-flow between $s_i,t_i$ and the
  flow from $s_i$ to $v$. It is easy to see that overlaying all of the
  flows will cause capacities to be violated by a factor of at most
  $3$; we scale down the flows by a factor of $3$ to satisfy the
  capacity constraints. Via this process we can find a new fractional
  solution $\bx'$ such that (1) all the flow paths contain $v$ and (2)
  $\sum_i x'_i \ge \frac{\sum_i y_i}{3} \ge \frac{\sum_i x_i}{3\alpha |S|}$.
  The result now follows by applying Proposition~\ref{prop:oldsinglenode}.
\end{proof}

\section{MEDP in Bounded Genus Graphs}
\label{sec:genus}
In this section we consider \medp and obtain an approximation ratio
that depends on the genus of the given graph. Throughout we use $g$ to
denote the genus of the given graph $G$; we   assume that
$g > 0$ since we already understand planar graphs. We assume that
the given instance of \medp is a matching instance in which
the terminals are degree $1$ leaves and that each terminal participates
in exactly one pair. An instance of \medp with this restriction is
characterized by a tuple $(G,X,M)$ where $G$ is the graph and $X$ is
the set of terminals and $M$ is a matching on the terminals
corresponding to the given pairs.  We  also assume that the degree
of each node is at most $4$; this can be arranged without changing the
genus by replacing a high-degree node by a grid (see \cite{Chekuri04b} for the
description for planar graphs which generalizes easily for any
surface).  Let $\bx$ be a feasible fractional solution to the
multicommodity flow based linear programming relaxation.  We use again the
terminology $\val{x}$ to denote $\sum_i x_i$, the fractional
amount of flow routed by $\bx$.  We let $\beta_g$ denote the flow-cut
gap for product-multicommodity flow instances in a graph of genus $g$;
this is known to be $O(\log (g+1))$ \cite{lee2010genus}. We  implicitly
assume that $\beta_g$ is an effective upper bound on the flow-cut gap
in that there is a polynomial-time algorithm that outputs a sparse cut
no worse than $\beta_g$ times the maximum concurrent flow for
a given product multicommodity instance on a genus $g$ graph.
The main result is the following.

\begin{theorem}
  \label{thm:genus_main}
  Let $\bx$ be a feasible fractional solution to a \medp instance in a
  graph of genus $g > 0$. Then
  $\Omega(\val{\bx}/\gamma_g)$ pairs can be routed with congestion
  $3$ where $\gamma_g = O(g \log^2 (g+1))$.
\end{theorem}

As we remarked previously we do not have currently have a polynomial-time
algorithm for the guarantee that we establishes in the preceding theorem.
There are three high-level ingredients in establishing the preceding theorem.
\begin{itemize}
\item A constant factor approximation with congestion $2$ for planar
  graphs based on the LP relaxation \cite{CKS-planar-constant,SeguinS11}.
\item An $O(g)$-approximation with congestion $3$ for graphs with
  genus $g$ when the terminals are well-linked. This extends the
  result in \cite{CKS-well-linked} for planar graphs to
  bounded genus graphs via the use of grid minors.
\item A modification of the well-linked decomposition of
  \cite{CKS-well-linked} that terminates the decomposition
   when the graph is planar even if the terminals are not well-linked.
\end{itemize}

We first give some formal definitions on well-linked sets;
the material follows \cite{CKS-well-linked} closely.

\medskip
\noindent {\bf Well-linked Sets:} Let $X \subseteq V$ be a set of
nodes and let $\pi : X \rightarrow [0,1]$ be a weight function on $X$.
We say $X$ is {\em $\pi$-flow-well-linked} in $G$ if there is a
feasible multicommodity flow in $G$ for the following demand matrix:
between every unordered pair of terminals $u,v \in X$ there is a
demand $\pi(u)\pi(v)/\pi(X)$ (other node pairs have zero demand).  We
say that $X$ is {\em $\pi$-cut-well-linked} in $G$ if $|\delta(S)| \ge
\pi(S \cap X)$ for all $S$ such that $\pi(S \cap X) \le \pi(X)/2$.  It
can be checked easily that if a set $X$ is $\pi$-flow-linked in $G$,
then it is $\pi/2$-cut-linked.  If $\pi(u) = \alpha$ for all $u \in
X$, we say that $X$ is {\em $\alpha$-flow(or cut)-well-linked}.  If
$\alpha=1$ we simply say that $X$ is well-linked.  Given $\pi:X
\rightarrow [0,1]$ one can check in polynomial time whether $X$ is
$\pi$-flow-well-linked or not via linear programming.

One can efficiently find an approximate sparse cut
if $X$ is not $\pi$-flow-well-linked via the algorithmic aspects
of the product flow-cut gap; the lemma below follows from
\cite{lee2010genus}.

\begin{lemma}
\label{lem:sparsecut}
Let $G=(V,E)$ be  a graph of genus at most $g > 1$.
Let $X \subseteq V$ and $\pi : X \rightarrow [0,1]$. There is
a polynomial-time algorithm that given $G$, $X$ and $\pi$ decides
whether $X$ is $\pi$-flow-well-linked in $G$ and if not outputs
a set $S \subseteq V$ such that $\pi(S) \le \pi(V\setminus S)$ and
$|\delta(S)| \le \beta_g \cdot \pi(S)$ where $\beta_g = O(\log (g+1))$.
\end{lemma}


We now formally state the theorems that correspond to the high-level
ingredients. The first is a constant factor approximation for
routing in planar graphs.

\begin{theorem}[\cite{CKS-planar-constant,SeguinS11}]
  \label{thm:planar-constant}
  Let $\bx$ be a feasible fractional solution to a \medp instance in
  a planar graph. Then there is a polynomial-time algorithm that routes
  $\Omega(\val{\bx})$ pairs with congestion $2$.
\end{theorem}

The second ingredient is an $O(g)$-approximation if the terminals
are well-linked. More precisely we have the following theorem which
we prove in Section~\ref{subsec:well-linked-genus}.
\begin{theorem}
  \label{thm:well-linked-genus}
  Let $\bx$ be a feasible fractional solution to a \medp instance in a
  graph of genus $g > 0$. Moreover, suppose the terminal set $X$ is
  $\pi$-flow-well-linked in $G$ where $\pi(v) = \rho \cdot x(v)$ for some
  scalar $\rho \le 1$. Then there is an algorithm
  that routes $\Omega(\rho \cdot \val{\bx}/g)$ pairs with congestion
  $3$.
\end{theorem}

The last ingredient is the following adaptation of the well-linked
decomposition from \cite{CKS-well-linked}.
\begin{theorem}
  \label{thm:genus-decomp}
  Let $(G,X,M)$ be an instance of \medp in a graph of genus at most
  $g$ and let $\bx$ be a feasible fractional solution. Then there is a
  polynomial-time algorithm that decomposes the given instances
  into several instances $(G_1,X_1,M_1),\ldots,(G_h,X_h,M_h)$
  with the following properties.
  \begin{itemize}
  \item The graphs $G_1,G_2,\ldots,G_h$ are node-disjoint subgraphs of $G$.
  \item For $1 \le j \le h$, $M_j \subseteq M$ and the end points of
    $M_j$ are in $G_j$.
  \item For $1 \le j \le h$, there is a feasible fractional solution
    $\bx^j$ for the instance $(G_j,X_j,M_j)$ such that $\sum_j \val{\bx^j}
    = \Omega(\val{\bx})$.
  \item For $1 \le j \le h$, $G_i$ is planar or the terminals $X_j$
    are $\bx^j/(10\beta_g \log (g+1))$-flow-well-linked  in $G_j$.
  \end{itemize}
\end{theorem}

Assuming the preceding three theorems we finish the proof of
Theorem~\ref{thm:genus_main}. Let $G$ be a graph of genus $\le g$.  We
apply the decomposition given by Theorem~\ref{thm:well-linked-genus}
which reduces the original problem instance $(G,X,M)$ to a collection
of separate instances $(G_1,X_1,M_1),\ldots,(G_h,X_h,M_h)$.  Note that
pairs in these new instances are from the original instance and
routings in these instances can be combined into a routing in the
original graph $G$ since the $G_1,G_2,\ldots,G_h$ are node-disjoint
(and hence also edge-disjoint). The total fractional solution value in
the new instances is $\Omega(\val{\bx}/\beta_g \log (g+1))$. If $G_j$
is planar we use Theorem~\ref{thm:planar-constant} to route
$\Omega(\val{\bx^j})$ pairs from $M_j$ in $G_j$ with congestion
$2$. If $G_j$ is not planar, then $X_j$ is $\bx^j/(10\beta_g \log
(g+1))$-flow-well-linked. Then, via
Theorem~\ref{thm:well-linked-genus}, we can route $\Omega(\val{\bx^j}/(g
\beta_g \log (g+1)))$ pairs from $M_j$ in $G_j$ with congestion $4$. Thus
we route in total $\Omega(\sum_j \val{\bx_j}/(g \beta_g \log (g+1))$ pairs.
Since $\sum_j \val{\bx_j} = \Omega(\val{x})$, we route a total of
$\Omega(\val{\bx}/(g \beta_g \log (g+1))) = \Omega(\val{\bx}/(g \log^2 (g+1)))$
pairs from $M$ in $G$ with congestion $3$.

\subsection{Proof of Theorem~\ref{thm:genus-decomp}}
\label{subsec:genus-decomp}

We start with a well-known fact.

\begin{prop}
  \label{prop:genus}
  Let $G=(V,E)$ be a connected graph of genus $g$. Let $S \subset V$
  be such that $G_1 = G[S]$ and $G_2=G[V\setminus S]$ are both connected.
  Then $g_1+g_2 \le g$ where $g_1$ is the genus of $G$ and $g_2$ is the
  genus of $G_2$.
\end{prop}

The algorithm below is an adaptation of the well-linked decomposition
algorithm from \cite{CKS-well-linked} that recursively partitions a
graph if the terminals are not well-linked (with a certain
parameter). In the adaptation below we stop the partitioning if the
graph becomes planar even if the terminals are not well-linked. We
start with an instance $(G,X,M)$ and an associated fractional solution
$\bx$. We fix a particular selection of flow paths for the solution $\bx$.
The algorithm recursively cuts $G$ into subgraphs by removing edges.
The flow paths that use the removed edges are lost and each connected
component retains flow corresponding to those flow paths which are
completely contained in that component.

\smallskip
\noindent
{\bf Decomposition Algorithm:}

\begin{enumerate}
\item \label{alg:term1} If $\val{\bx} < 10 \beta_g \log (g+1)$ or if $G$ is
  a planar graph stop and output $(G,X,M)$ with fractional solution
  $\bx$.

\item Else if $X$ is $\bx/(10 \beta_g \log (g+1))$ flow-well-linked in $G$ then
  stop and output $(G,X,M)$ with fractional solution $\bx$.

\item Else find a sparse cut $\delta(S)$ such that $x(S) \le
  \val{\bx}/2$ and $|\delta(S)| \le 1/(10 \log (g+1)) \cdot x(S)$. Let
  $G_1=G[S]$ and $G_2=G[V\setminus S]$. (Assume wlog that $G_1,G_2$
  are connected).  Let $(G_1,X_1,M_1) $ and $(G_2,X_2,M_2)$ be the
  induced instances on $G_1$ and $G_2$ with fractional solutions $\bx^1$
  and $\bx^2$ respectively.  Recurse
  separately on $(G_1,X_1,M_1) $ and $(G_2,X_2,M_2)$.
\end{enumerate}

We now prove that the above algorithm outputs a decomposition as
stated in Theorem~\ref{thm:genus-decomp}.  The only property that is
non-trivial to see is the one about the total flow retained in the
decomposition. As in \cite{CKS-well-linked} it is easier to upper
bound the total number of edges cut by the algorithm which in turn
upper bounds the total amount of flow from the original solution $\bx$
that is lost during the decomposition. We write a recurrence for this
as follows.  Let $L(a,r)$ be total number of edges cut by the
algorithm if $a$ is the amount of flow in the graph and $r\le g$ is
the genus. The base case is $L(a,0) = 0$ since the algorithm stops
when the graph is planar. The algorithm cuts and recurses only when
the terminals are not $\bx/(10 \beta_g \log (g+1))$ flow-well linked.
Suppose $H$ is the current graph with flow value $a = \sum_v x(v)/2$
and genus $r$, and $H$ is partitioned into $H_1$ and $H_2$. Let $a_1$
and $a_2$ be the total flow values in $H_1$ and $H_2$ and let
$r_1$ and $r_2$ be their genus respectively. We have $a_1+a_2 \le a$ and
by Proposition~\ref{prop:genus} we have $r_1+r_2 \le r$.  We claim
that the number of edges cut is at most $\frac{1}{4 \log (g+1)}
\min\{a_1,a_2\}$. To see this, suppose $\delta(S)$ is the sparse cut
in $H$ that resulted in $H_1$ and $H_2$ with $H_1=H[S]$ and $H_2 =
H[V\setminus S]$ and $a_1 = \min\{a_1,a_2\}$. Let $\bx$ be the
fractional solution in $H$ and $\bx'$ the solution after the
partitioning. We have $a = x(V(H))/2$. Since the terminals are not
$\bx/(10 \beta_g \log (g+1))$ flow-well-linked in $H$, by
Lemma~\ref{lem:sparsecut}, $|\delta(S)| \le \beta_g \cdot x(S)/(10
\beta_g \log (g+1)) \le x(S)/(10 \log (g+1))$. We have $a_1 =
x'(S)/2$. Moreover, $x'(S) \ge x(S) - 2 |\delta(S)|$ since a flow path
$p$ crossing the cut with flow $f_p$ can contribute at most $2f_p$ to
the reduction of the marginal values of $x(S)$ at its end points.
Thus $2a_1 = x'(S) \ge x(S) -2 |\delta(S)| \ge (10 \log (g+1) - 2)
|\delta(S)| \ge 8 |\delta(S)|$ since $g \ge 1$, which implies
that $|\delta(S)| \le a_1/4$, as claimed.

Therefore we have the following recurrence for the total
number of edges cut in the overall decomposition:

$$L(a,r) \le L(a_1,r_1) + L(a_2,r_2) + \frac{1}{4 \log (g+1)} \min\{a_1,a_2\}.$$

We prove by induction on $r$ and number of nodes of $G$ that for $r
\le g$, $L(a,r) \le \frac{\log (r+1)}{2\log (g+1)} \cdot a$. We had
already seen the base case with $r = 0$ since $L(a,r) = 0$.  By
induction $L(a_1,r_1) \le \frac{\log (r_1+1)}{2\log (g+1)} \cdot a_1$
and $L(a_2,r_2) \le \frac{\log (r_2+1)}{2\log (g+1)} \cdot a_2$.
Since $r_1+r_2 \le r$, $\min\{r_1+1,r_2+1\} \le (r+1)/2$. It is not
hard to see that $a_1 \log (r_1+1) + a_2 \log (r_2+1) \le (a_1+a_2)
\log (r+1) - \min\{a_1,a_2\} \log 2$. Using the recurrence.
\begin{eqnarray*}
  L(a,r) & \le & L(a_1,r_1) + L(a_2,r_2) + \frac{1}{4 \log (g+1)} \min\{a_1,a_2\} \\
  & \le & \frac{\log (r_1+1)}{2\log (g+1)} \cdot a_1 + \frac{\log (r_2+1)}{2\log (g+1)} \cdot a_2 + \frac{1}{4 \log (g+1)} \min\{a_1,a_2\} \\
  & \le &  \frac{\log (r+1)}{2 \log (g+1)} (a_1 + a_2)  - \frac{\log 2}{2\log (g+1)} \min\{a_1,a_2\} + \frac{1}{4 \log (g+1)} \min\{a_1,a_2\} \\
  & \le & \frac{\log (r+1)}{2 \log (g+1)} \cdot a.
\end{eqnarray*}

Thus $L(\val{\bx},g) \le \frac{\log (g+1)}{2 \log (g+1)} \val{\bx} \le
\val{\bx}/2$.  Hence the total flow that remains after the
decomposition is at least $\val{\bx}/2$.

\subsection{Proof of Theorem~\ref{thm:well-linked-genus}}
\label{subsec:well-linked-genus}

We start with a grouping technique from \cite{CKS-well-linked} that
boosts the well-linkedness.

\begin{theorem}
  \label{thm:boost}  Let $\bx$ be a feasible fractional solution to an
  \medp instance $(G,X,M)$.  Moreover, suppose the terminal set $X$ is
  $\pi$-flow-well-linked in $G$ where $\pi(v) = \rho \cdot x(v)$ for
  some scalar $\rho \le 1$. Then there is a polynomial-time algorithm
  that routes $\Omega(\rho \cdot \val{\bx})$ pairs of $M$ edge-disjointly or
  outputs a new instance $(G,X',M')$ where $M' \subset M$ and $X'$ is
  well-linked and $|M'| = \Omega(\rho |M|)$.
\end{theorem}

Using the preceding theorem we assume that we are working with an
instance $(G,X,M)$ where $X$ is well-linked. Recall that $G$ has genus
at most $g$ and degree of each node is at most $4$.  We use the
observation below (see \cite{Chekuri04b}) that relates the
size of a well-linked set and the treewidth.
\begin{lemma}
  Let $G$ be a graph with maximum degree $\Delta$ and let $X \subseteq V$ be
  a well-linked set in $G$. Then the treewidth of $G$ is $\Omega(|X|/\Delta)$.
\end{lemma}

Thus we can assume that $G$ has treewidth $\Omega(|X|)$.  Demaine
\etal \cite{Demaineetal-bidimensional} showed the following theorem on the size of a grid minor in
graphs of genus $g$ following the work of Robertson and Seymour.

\begin{theorem}[\cite{Demaineetal-bidimensional}]
  Let $G$ be a graph of genus at most $g$. Then $G$ has a grid minor
  of size $\Omega(h/g)$ where $h$ is the treewidth of $G$.
\end{theorem}

Following the scheme from \cite{CKS-well-linked}, one can use the grid
minor as a cross bar to route a large number of pairs from $M$
provided we can route $\Omega(|X|/g)$ terminals to the ``interface''
of the grid-minor of size $\Omega(|X|/g)$ that is guaranteed to exist
in $G$. We can view the grid minor as rows and columns.  In our
current context we take every other node in the first row of the
grid as the interface of the grid-minor. Each node $v$ of the minor
corresponds to a subset of connected nodes $A_v$ in original graph $G$ that
are contracted to form $v$. To simplify
notation we say that $S \subset V$ is the interface of a grid-minor
in $G$ if $|S \cap A_v| = 1$ for each interface node $v$ of the
grid-minor. The following lemma is essentially implicit in previous
work, in particular \cite{CKS-planar} used in the context of planar
graphs.

\begin{lemma}
  \label{lem:grid-routing}
  Suppose $G=(V,E)$ contains a $h \times h$ grid as a minor. Let $S
  \subseteq V$ be the interface of the grid-minor. Then $S$ is
  well-linked in $G$. Moreover, any matching $M$ on $S$ is routable
  in $G$ with congestion $2$.
\end{lemma}

The key technical difficulty is to ensure that $G$ contains
a large grid-minor whose interface is reachable from the terminals $X$.
In a sense, the grid-minor's existence is shown via the
well-linkedness of $X$ and hence there should be such a ``reachable''
grid-minor. The following theorem formalizes the existence
of the desired grid-minor.

\begin{theorem}
  \label{thm:reachable-grid}
  Let $G$ be a graph of genus $g > 0$. Suppose $X$ is a well-linked
  set in $G$. Then $G$ contains a $h \times h$ grid-minor
  with interface $S$ such that $h = \Omega(|X|/g)$ and
  at least $h/8$ edge-disjoint paths in $G$ from $X$ to $S$.
\end{theorem}

Proving the preceding theorem formally 
requires work. In \cite{CKS-planar} an argument tailored to planar
graphs was used to prove a similar theorem, and in fact a
polynomial-time algorithm was given to find the desired grid-minor. In that
same paper a general deletable edge lemma \cite{cks04} was announced (though
never published); in the appendix we give  a streamlined
proof based on the original manuscript. We
postpone the proof of the preceding theorem and outline how to
complete the proof of Theorem~\ref{thm:well-linked-genus}.
We start with our well-linked set $X$, and via Theorem~\ref{thm:reachable-grid},
find a grid-minor of size $h = \Omega(|X|/g)$ such that there
are $h/8$ edge-disjoint paths from $X$ to the interface $S$.
Let $X' \subset X$ be $h/8$ terminals that are the end points of these
edge-disjoint paths. Recall that we are interested in routing
a given matching $M$ on $X$. If $X'$ contains a ``large'' sub-matching
$M' \subset M$ then we can use the grid-minor to route $M'$ with
congestion $3$ as follows. Let $S' \subset S$ be end points of
the paths from $X'$ to $S$. Clearly $M'$ induces a matching on
$S'$ and the grid-minor can route this matching with congestion $2$.
Patching this routing with the paths from $X'$ to $S'$ gives the
desired congestion $3$ routing of $M'$ in $G$. However, it may be the
case that $|M'|$ is very small, even zero,
because $X'$ only contains one end point from every edge in $M$.
However, here, we can use the fact that $X$ and $S$
are  well-linked in $G$ to argue that we can route any desired
subset of $X$ of sufficiently large size to $S$. Thus, we can assume
that indeed $|M'|$ is a constant fraction of $|X'|$.
This was essentially done in \cite{Chekuri04b} and details
are also included in the appendix (see Lemma~\ref{lem:alltoall2}).
This shows the number of pairs from $M$ that can be routed
with congestion $3$ in $G$ is a constant fraction of $h$, the size
of the grid-minor. Since $h = \Omega(|X|/g)$ we route
$\Omega(|X|/g)$ pairs from $M$. This finishes the proof of
Theorem~\ref{thm:well-linked-genus}.


Now we come to the proof of Theorem~\ref{thm:reachable-grid}.  This is
done in Section~\ref{sec:deletable-edge} via the following high-level approach.  Suppose
$X$ is a well-linked set in $G$. We are guaranteed that $G$ has a
grid-minor of size $h = \Omega(|X|/g)$. Let $S$ be its interface. If
there are $h/8$ edge-disjoint paths from $X$ to $S$ then we are
done. Otherwise the claim is that there is an edge $e$ such that $X$
is well-linked in $G-e$. This is established in
Theorem~\ref{thm:deletable}.  In other words, if we start with a graph
which is edge-minimal subject to $X$ being well-linked, then following
the procedure above  yields a grid-minor that is reachable from
$X$. The reason that this does not immediately lead to a polynomial-time
algorithm to find such a grid-minor is the fact that checking whether
a given set $X$ is not well-linked is NP-Complete. Note that we can
check if $X$ if flow-well-linked but the deletable edge lemma we have
works with the notion of cut-well-linkedness which is hard to check.
In \cite{CKS-planar} a poly-time deletable edge lemma was obtained
for the special case of planar graphs and we believe that it can be
extended to the case of genus $g$ graphs but the technical details
are quite involved.

\section{Open Problems and Concluding Remarks}
Resolving Conjecture~\ref{conj:minor-free} is the main open problem
that arises from this work. In particular, does a planar graph with a
single vortex satisfy the CFCC property? This is
is the key technical obstacle.

Theorem~\ref{thm:ksum-main} loses an approximation factor that is
exponential in $k$. Is this necessary? In particular, the known
integrality gap for the flow LP on graphs of treewidth $k$ is only
$\Omega(k)$; this comes from the grid example. Moreover, the
congestion we obtain for $\cG_k$ is $\beta+3$ where $\beta$ is the
congestion guaranteed for the class $\cG$. Can this be improved
further, say to $\beta$ or $\beta+1$?

We believe that Theorem~\ref{thm:genus_main} can be made algorithmic
and also conjecture that the congestion bound can be improved to
$2$. The key bottleneck is to find an algorithmic proof of
Theorem~\ref{thm:reachable-grid}.

\fi

\section{A Deletable Edge Lemma }
\label{sec:deletable-edge}
The following results are from the unpublished note \cite{cks04}. We
have streamlined the original proofs and include it for completeness.

Throughout this section we use well-linked to mean cut-well-linked.
We say that $S \subseteq V$ is {\em routable} to $T \subseteq V$ in
$G$ if there are $|S|$ edge-disjoint paths from $S$ to $T$ in $G$ such
that each node in $S \cup T$ is the end point of at most one of the
paths. We allow paths of length $0$ if for example a node belongs to
$S \cap T$. If $S$ can be routed to $T$ then of course $|S| \le |T|$.

Given a graph $G$, two sets of nodes $A$ and $B$, we define an {\em
  auxiliary graph} $G(A,B)$ obtained by attaching a node $s_A$ (or
just $s$) and a node $t_B$ (or just $t$). The node $s$ has an edge to
each node in $A$ and $t$ has an edge to each node in $B$. The max
$s$-$t$ flow in $G(A,B)$ identifies the maximum routable set from $A$
to $B$.

\begin{prop}
  \label{prop:linked}
  If $H$ is well-linked in $G$ then for any $X,Y \subset H$ and
  $|X| \le |Y|$, $X$ is routable to $Y$ in $G$.
\end{prop}
\begin{proof}
  Consider the auxiliary graph $G':=G(X,Y)$.   It
  is sufficient to prove that the $s$-$t$ mincut in $G'$ is $|X|$. Let
  $\delta(S')$ be an $s$-$t$ minimum cut in $G'$ with $s \in S'$.  Let $S = S'
  - \{s\}$. We assume wlog that $|S \cap H| \le |H|/2$, otherwise we
  can work with $V(G') \setminus S$. Since $H$ is well-linked it
  follows that $|\delta_G(S)| \ge |S \cap H|$. We have that
  $|\delta_{G'}(S')| \ge |X \setminus S| + |\delta_G(S)| \ge |X
  \setminus S| + |S \cap H| \ge |X \setminus S| + |S \cap X| \ge |X|$.
\end{proof}

\begin{lemma}
  \label{lem:alltoall1}
  Let $H_1$ and $H_2$ be two disjoint well-linked sets in $G$. Suppose $A
  \subset H_1$ is routable to $B \subset H_2$ and $|A| \le |H_1|/2$. Then
  given any $A' \subset H_1$ with $A' \le |A|/2$, $A'$ is routable to $B$.
\end{lemma}
\begin{proof}
  Consider the auxiliary graph $G':=G(A',B)$.  Let
  $\delta_{G'}(S')$ be an $s$-$t$ minimum cut in $G'$ with $s \in S'$. Let $S = S'
  - \{s\}$. We argue that $|\delta_{G'}(S')| \ge |A'|$ which
  proves the lemma.  Let $a=|A|$. We have the equality that

  \begin{equation}
    \label{eqn:cutdefn}
    |\delta_{G'}(S')| = |\delta_G(S)| + |A' \setminus S| + |B \cap S|.
  \end{equation}

  We consider two cases.  In the first case $|S \cap A| \ge
  a/2$. Since $A$ is routable to $B$ it follows that $|\delta_G(S)| \ge
  a/2 - |S \cap B|$. Hence from Equation~\ref{eqn:cutdefn} we have
  that $|\delta_{G'}(S')| \ge a/2 \ge |A'|$.

  In the second case, $|S \cap A| < a/2$ which implies that $|(V-S)
  \cap A| \ge a/2$. Let $Y = (V-S) \cap A$. Therefore $|Y| \ge a/2$.
  By Proposition~\ref{prop:linked} we have that $S \cap A'$ is
  routable to $Y$.  It follows that $|\delta_G(S)| \ge |S \cap A'|$.
  From Equation~\ref{eqn:cutdefn}
  $|\delta_{G'}(S')| \ge |S \cap A'| + |A' \setminus S| = |A'|$.
\end{proof}

\begin{lemma}
  \label{lem:alltoall2}
  Let $H_1$ and $H_2$ be two disjoint well-linked sets in $G$. Suppose
  $A \subset H_1$ is routable to $B \subset H_2$. Then given any $A'
  \subset H_1$ and $B' \subset H_2$ with $|A'| = |B'| \le |A|/3$, $A'$
  is routable to $B'$.
\end{lemma}
\begin{proof}
  Consider the auxiliary graph $G':=G(A',B')$.
  Let $\delta_{G'}(S')$ be an $s$-$t$ minimum cut in $G'$ with $s \in
  S'$. Let $S = S' - \{s\}$ and $T = V(G)-S$.  We argue that
  $|\delta_{G'}(S')| \ge |A'|$ which  proves the lemma.  Let $a=|A|$.
  We have the equality that

  \begin{equation}
    \label{eqn:cutdefn2}
    |\delta_{G'}(S')| = |\delta_G(S)| + |A' \setminus S| + |B' \cap S|.
  \end{equation}

  We can assume that $|B| = |A|$. Consider a fixed routing from $A$ to
  $B$.  We call $u \in A, v \in B$ a {\em pair} if $u$ and $v$ are joined by
  a path in the fixed routing. Suppose the cut $\delta_G(S)$ separates at
  least $a/3$ pairs.  Then clearly $|\delta_G(S)| \ge a/3$ and we are
  done. Therefore either $S$ contains at least $a/3$ pairs or $T$
  contains $a/3$ pairs.  If $S$ contains at least $a/3$ pairs, then $S$
  contains $a/3$ nodes from $H_2$. Then by the
  well-linkedness of $H_2$ we see that $|\delta_G(S)| \ge |T \cap B'|$
  and hence $|\delta_{G'}(S')| \ge |T \cap B'| + |B' \cap S| = |B'|$.
  If $T$ contains at least $a/3$ pairs, then we can use the well-linkedness
  of $H_1$ to argue that $|\delta_{G'}(S')| \ge |A'|$.
\end{proof}

We note that Lemmas~\ref{lem:alltoall2} and \ref{lem:alltoall2} are
tight.

Let $H_1$ and $H_2$ be two disjoint well-linked sets in $G$.  Let
$\MC$ be the size of a min-cut $s$-$t$ cut in the auxiliary graph
induced by $H_1$ and $H_2$.  Thus there is a subset $A$ of $H_1$, and
subset $B$ of the $H_2$ such that $|A|=|B|=\MC$ and $A$ is routable to
$B$.

The main result in this section is the following:
\begin{theorem}
  \label{thm:deletable}
  If $\gamma < |H_2|/8$, then there is an edge $e$ in $G$ such
  that $H_1$ is well-linked in $G - e$.
\end{theorem}

For simplicity we assume that each node in $H_1$ has degree one in
$G$. This is without loss of generality. For each node $v \in H_1$
we can add a dummy node $v'$ and attach $v'$ to $v$ by a single
edge; the set $H'_1 = \{v' | v \in H_1\}$ is easily seen to be
well-linked and the routability of $H'_1$ to $H_2$ is the same
as that from $H_1$ to $H_2$. Further, none of the edges $(v',v)$
is deletable and hence any edge that we show is deletable in the
modified instance is an edge in the original graph.

Let $k = |H_2|$. To prove the above theorem we consider an $s$-$t$
minimum cut $\delta(S')$ in the auxiliary graph $G':=G(H_1,H_2)$.
Let $S = S' - s$ and $T = V(G)-S$.  We can evidently choose such a cut so that
$|T \cap
H_1| = 0$ since each node in $H_1$ is a leaf. We also have that
that $|T \cap H_2| \ge |H_2| - \gamma \ge |H_2|/2$.  Given the existence of $T$,
we may now choose a
$\MC$-cut $\delta(\C)$ in $G$ with $\C$ minimal subject to satisfying
the property that $|\C \cap H_2| \ge |H_2|/2 = k/2$ and $|\C \cap H_1| =
0$. Note that $\C$ is not a stable set for otherwise
$|\delta(\C)| \ge k/2$ but by definition $|\delta(\C)| = \gamma$.

  \begin{figure}[htb]
    \centering
    \includegraphics[width=3.5in]{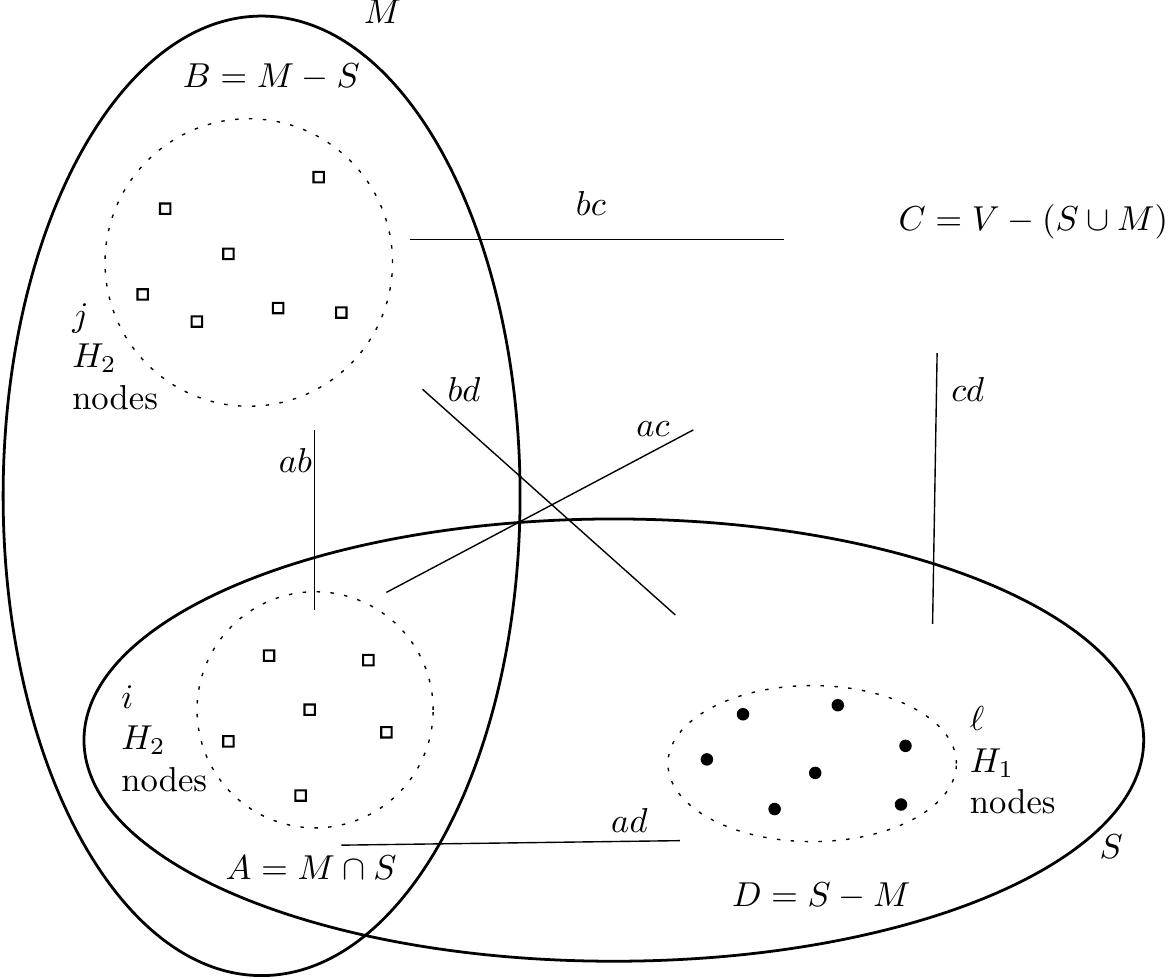}
    \caption{Illustration for proof of removable edge from $M$.}
    \label{fig:remove}
  \end{figure}

\begin{lemma}
  \label{lem:deletable}
  Any edge with both ends in $\C$ is deletable.
\end{lemma}
\begin{proof}
  Let $e$ be an edge inside $M$. Suppose $e$ is not deletable. Then
  there is a light set $S$ with respect to $H_1$ that is tight and
  $e$ crosses $S$. That is, $|S \cap H_1| = \ell \le |H_1|/2$ and
  $|\delta(S)| = \ell$ and $e \in \delta(S)$. Let $i$ and $j$ be the
  number of $H_2$ nodes in $M \cap S$ and $M -  S$
  respectively. Recall that $M$ does not contain any $H_1$ nodes.
  See Fig~\ref{fig:remove}.

  We first observe that $(i+j) \ge k - \gamma \ge 7k/8$, otherwise more than
  $\gamma$ nodes are in $V -  M$ and $|\delta(V -  M)|
  =|\delta(M)| = \gamma$; this would violate the well-linkedness of $H_2$.
  Second, by submodularity and symmetry of the function $|\delta_G|:2^V \rightarrow \mathbb{Z}_+$, we have,
  $$|\delta(M)| + |\delta(S)| \ge |\delta(M \cap S)| + |\delta(M \cup S)|,$$ and
  $$ |\delta(M)| + |\delta(S)| \ge |\delta(M -  S)| + |\delta(S  -  M)|.$$

  We have $|\delta(M)| = \gamma$ and $|\delta(S)| = \ell$. Both $M
  \cup S$ and $S-M$ have exactly $\ell$ nodes of $H_1$; since $H_1$ is
  well-linked, $|\delta(M \cup S)| \ge \ell$ and $|\delta(S-M)| \ge
  \ell$.  Thus, from the above submodularity inequalities, we get a
  contradiction if we can prove that $|\delta(M \cap S)| > \gamma$
  {\em or} $|\delta(S - M)| > \gamma$. If $i \ge k/2$ we have
  $|\delta(M \cap S)| > \gamma$ for otherwise $M \cap S$ contradicts
  the minimality of $M$. Similarly if $j \ge k/2$ we have
  $|\delta(M-S)| > \gamma$ for otherwise $M - S$ contradicts the
  minimality of $M$.

  We now consider the case that $i < k/2$ {\em and} $j < k/2$.  Since
  $(i+j) \ge 7k/8$ we have $7k/16 \le \max\{i,j\} < k/2$.  If $i =
  \max\{i,j\}$ then by well-linkedness of $H_2$, $|\delta(M \cap S)|
  \ge i \ge 7k/16 > \gamma$.  If $j = \max\{i,j\}$ then again by
  well-linkedness of $H_2$, $|\delta(M-S)| \ge j \ge 7k/16 >
  \gamma$. In both cases we get the desired contradiction.
\end{proof}

\end{document}